\documentclass[11pt,onecolumn]{IEEEtran}
\usepackage[tbtags]{amsmath}
\usepackage{amsfonts}
\usepackage{amssymb}
\usepackage{amsthm}
\usepackage{algorithmicx}
\usepackage{algpseudocode}
\usepackage{algorithm}
\usepackage{subfigure}
\usepackage{graphicx}
\usepackage{cite}
\usepackage{calc}
\usepackage{color}
\usepackage{epsfig}
\usepackage{enumitem}
\usepackage{setspace}
\usepackage{multirow}
\usepackage{hyperref}
\usepackage{mathtools, cuted}
\usepackage{epstopdf}
\usepackage{lipsum}
\usepackage{mathtools}
\usepackage{cuted}
\usepackage[thinlines]{easytable}

\makeatletter
\def\endthebibliography{%
	\def\@noitemerr{\@latex@warning{Empty `thebibliography' environment}}%
	\endlist
}
\makeatother

\makeatletter
\def\footnoterule{\kern-3\p@
	\hrule \@width 2in \kern 2.6\p@} 
\makeatother

\newtheorem{defn}{Definition}
\newtheorem{thm}{{\cal T}heorem}[section]
\newtheorem{cor}[thm]{Corollary}
\newtheorem{prop}{Proposition}
\newtheorem{lem}[thm]{Lemma}
\newtheorem{conj}[thm]{Conjecture}
\newtheorem{constr}[thm]{Construction}
\newtheorem{note}{Remark}
\newcommand{\bit}{\begin{itemize}}
\newcommand{\eit}{\end{itemize}}
\newcommand{\bcor}{\begin{cor}}
\newcommand{\ecor}{\end{cor}}
\newcommand{\beq}{\begin{equation}}
\newcommand{\eeq}{\end{equation}}
\newcommand{\beqn}{\begin{equation}}
\newcommand{\eeqn}{\end{equation}}
\newcommand{\bea}{\begin{eqnarray}}
\newcommand{\eea}{\end{eqnarray}}
\newcommand{\bean}{\begin{eqnarray*}}
\newcommand{\eean}{\end{eqnarray*}}
\newcommand{\ben}{\begin{enumerate}}
\newcommand{\een}{\end{enumerate}}
\newcommand{\bdefn}{\begin{defn}}
\newcommand{\edefn}{\end{defn}}
\newcommand{\bnote}{\begin{note}}
\newcommand{\enote}{\end{note}}
\newcommand{\bprop}{\begin{prop}}
\newcommand{\eprop}{\end{prop}}
\newcommand{\blem}{\begin{lem}}
\newcommand{\elem}{\end{lem}}
\newcommand{\bthm}{\begin{thm}}
\newcommand{\ethm}{\end{thm}}
\newcommand{\bconj}{\begin{conj}}
\newcommand{\econj}{\end{conj}}
\newcommand{\bconstr}{\begin{constr}}
\newcommand{\econstr}{\end{constr}}
\newcommand{\bpf}{\begin{proof}}
\newcommand{\epf}{\end{proof}}
\newcommand{\bprf}{{\em Proof: }}
\newcommand{\eprf}{\hfill $\Box$}

\newcommand{\cxyz}{\mbox{$C(x,y; \underline{z})$}}
\newcommand{\cxyzc}{\mbox{$C(z_y,y; \underline{z}(x \rightarrow z_y))$}}
\newcommand{\cxyznot}{\mbox{$C(x_0,y_0; \underline{z})$}}
\newcommand{\cxyzcnot}{\mbox{$C(x_0,y_0; \underline{z}(x \rightarrow z_{y_0}))$}}
\newcommand{\calc}{\mbox{$\mathcal{C}$}}
\newcommand{\call}{\mbox{$\mathcal{L}$}}
\newcommand{\calf}{\mbox{$\mathcal{F}$}}

\newcommand{\uz}{\mbox{$\underline{z}$}}

\newcommand{\uw}{\mbox{$\underline{w}$}}

\newcommand{\Zs}{\mbox{$[0:s-1]$}}
\newcommand{\Zt}{\mbox{$[0:t-1]$}}
\newcommand{\Zst}{\mbox{$\mathbb{Z}_s^t$}}

\newcommand{\bc}{\begin{center}}
	\newcommand{\ec}{\end{center}}

\newcommand{\fq}{\ensuremath{\mathbb{F}_q}}

\newcommand{\uc}[1]{\ensuremath{\underline{c}_{#1}}}

\newcommand\scalemath[2]{\scalebox{#1}{\mbox{\ensuremath{\displaystyle #2}}}}

\begin{document}
\title{On Lower Bounds on Sub-Packetization Level of MSR codes and On The Structure of Optimal-Access MSR Codes Achieving The Bound}
\author{
\IEEEauthorblockN{S. B. Balaji, Myna Vajha, P. Vijay Kumar} \ \\
\IEEEauthorblockA{
	Department of Electrical Communication Engineering, IISc Bangalore \\ \{balaji.profess, mynaramana, pvk1729\}@gmail.com}
\footnote{This paper was presented in part at IEEE, International Symposium on Information Theory (ISIT) 2018 in \cite{BalKum}.}
}
\maketitle

\begin{abstract}
	We present two lower bounds on sub-packetization level $\alpha$ of MSR codes with parameters $(n, k, d=n-1, \alpha)$ where $n$ is the block length, $k$ dimension, $d$ number of helper nodes contacted during single node repair and $\alpha$ the sub-packetization level. The first bound we present is for any MSR code and is given by $\alpha \ge e^{\frac{(k-1)(r-1)}{2r^2}}$. 
	The second bound we present is for the case of optimal-access MSR codes and the bound is given by  $\alpha \ge \min \{ r^{\frac{n-1}{r}}, r^{k-1}  \}$. There exist optimal-access MSR constructions that achieve the second sub-packetization level bound with an equality making this bound tight.
   
    We also prove that for an optimal-access MSR codes to have optimal sub-packetization level under the constraint that the indices of helper symbols are dependant only on the failed node, it is needed that the support of the parity check matrix is same as the support structure of constructions shown in \cite{VajBalKum,SasAgaKum,SasVajKum,YeBarg_2017, RawKoyVis_msr}.
\end{abstract}


\section{Introduction}
%
Erasure codes are of strong interest in distributed storage systems as they offer reliability at lower values of storage overhead in comparison with replication.   In the setting of distributed storage, each code symbol is stored on a distinct storage unit, such as a hard disk.   Among the class of erasure codes, Maximum Distance Separable (MDS) codes are of particular interest as they offer reliability at lowest possible value of storage overhead.  An $[n,k]$ MDS code with block-length $n$ and dimension $k$ can recover from any $(n-k)$ erasures.

Apart from reliability and storage overhead a third important concern in a distributed storage system is that of efficient node repair.   Efficient node repair could either call for the amount of data download needed to repair a failed node to be kept to a low level or else, the number of helper nodes contacted for repair to be kept small.  Here the focus in on the first criterion, i.e., the amount of data download, which is also termed the repair bandwidth.    

\subsection{Minimum Storage Regenerating (MSR) Codes} Regenerating codes are erasure codes that offer node repair with least possible repair bandwidth for a fixed code alphabet and number of symbols encoded and block length and number of nodes contacted for repair.   Within the class of regenerating codes, the subclass of Minimum Storage Regenerating (MSR) codes are of particular interest, as the rate of an MSR code can be made as close to $1$ as possible.  In addition, MSR codes fall within the class of MDS codes and hence offer the best possible rate given their ability to recover from a fixed number of symbol erasures. An $[n,k]$ MSR code is an $[n,k]$ MDS codes over the vector alphabet $\fq^{\alpha}$ satisfying the additional constraint that a failed node can be repaired by contacting $d$ helper nodes, while downloading $\beta$ symbol over \fq\ from each helper node.  We use $B$ to denote the number $B=k \alpha$ of message symbls over \fq\ encoded by an MSR code.    Thus, MSR codes are characterized by the parameter set 
\bean
\left\{ (n,k,d),\ (\alpha,\beta),\ B, \ \fq \right\}, \text{ where }
\eean
\bit 
\item \fq\ is the underlying finite field, 
\item $n$ is the number of code symbols $\{\underline{c}_i\}_{i=1}^n$ each stored on a distinct node or storage unit and 
\item each code symbol $\underline{c}_i$ is an element of  $\fq^{\alpha}$. 
\eit

It turns out that the number $\beta$ of symbols downloaded from each helper node in an MSR code is given by 
\bean
\beta = \frac{\alpha}{d-k+1}. 
\eean
This is obtained by deriving minimum repair bandwidth $d \beta$ for the repair of a failed node in an MDS code over $\fq^{\alpha}$ for a fixed $n,k,d,\alpha$.
In a distributed storage system, each code symbol \uc{i} is typically stored on a distinct node. Thus the index $i$ of a code symbol is synonymous with the index of the node upon which that code symbol is stored. Throughout this paper, we will focus on a linear MSR code i.e., the encoding is done by:  $[\uc{1}^T,\hdots,\uc{n}^T]^T = G\underline{m}$ where $G$ is an $(n\alpha \times k\alpha)$ generator matrix over \fq\ and $\underline{m}$ is a $(k\alpha \times 1)$ message vector over \fq\ comprising of $B$ message symbols encoded by the MSR code.

\subsection{Desirable Properties of an MSR Code}

While MSR codes are MDS codes with smallest repair bandwidth possible in an MDS code, there is still scope for optimization with the class of MSR codes.  The additional features of interest are listed below.

\bit
\item {\em Optimal-Access:\ }
Optimal-access MSR codes \cite{TamoWangBruck} are a subclass of MSR codes having the property that during repair, the $\beta$ symbols that are transmitted by a helper node during repair, are simply a subset of the $\alpha$ symbols contained in the node.  This has two important and desirable, practical consequences.  Firstly, the number of symbols accessed in the node is. as small as possible and secondly, no computations are required to generate the transmitted repair symbols.  
\item {\em Low Values of Sub-Packetization:\ } When the MSR codes are implemented in systems, the finite field operations are done over an collection of a codewords. Let $J$ be the number of such codewords. Let $F$ be the size of file  that is being stored, then: $J = F / k\alpha.$
During repair $\beta J$ symbols are requested out of the $\alpha J$ symbols. However, they are not always stored sequentially. The minimum number of repaired symbols that are available sequentially is hence $J$ and may result in fragmented reads if $J$ is very small \cite{VajRamPur_Clay}. Therefore, larger the $J$, better is the access speed for these repaired symbols. This in turn results in need for lower value of sub-packetization level $\alpha$.
\item {\em Low Field Size:\ } The need for a low field size is clear since the smaller the size of the finite field, the lesser is the implementation complexity.  
\eit 


\subsection{Linear Repair of MSR Codes}  
Throughout this paper, we will assume linear repair of the failed node.   By linear repair, we mean that the $\beta$ symbols passed on from a helper node $i$ to the replacement of a failed node $j$ are obtained through a linear transformation:
\bean
\uc{i} \rightarrow S_{i \rightarrow j} \uc{i} ,
\eean
where $S_{i \rightarrow j}$ is an $(\beta \times \alpha)$ matrix over \fq.  While the matrix $S_{i \rightarrow j}$ could potentially be a function of which nodes are participating in the repair of failed node $j$, our paper is mostly concerned with the case $d=(n-1)$, in which case, all the remaining $(n-1)$ nodes participate as helper nodes. 

As a result, when $d=n-1$, the input to the replacement of failed node $j$ is the set:
\bean
\left\{   S_{i \rightarrow j} \uc{i}\mid    i \in [n], \ i \neq j             \right\} .
\eean 
By linear repair of a node $j$ we also mean: the code symbol $\uc{j} = f_j \left( \left\{   S_{i \rightarrow j} \uc{i}\mid    i \in [n], \ i \neq j  \right\} \right)$ where $f_j:\fq^{d\beta} \rightarrow \fq^{\alpha}$ is a deterministic linear function.
We refer to an MSR code as an {\em optimal access} MSR code if the matrix $S_{i \rightarrow j}$ has rows picked from the standard basis $\{e_1,\hdots,e_{\alpha}\}$ for $\fq^{\alpha}$., i.e., the symbols over \fq\ downloaded for the repair of a failed node from node $i$ are simply a subset of size $\beta$, of the $\alpha$ components of the vector $\underline{c}_i$. This property of repair is also termed as optimal-access repair as the number of symbols communicated for repair are exactly the symbols that are accessed. 

Since an $[n,k]$ MSR code is an MDS code, it can recover from the failure of any $r=(n-k)$ node failures or erasures.   Equivalently, the $B$ underlying data symbols encoded by the MSR code can be recovered form the contents of any subset of $k$ nodes.  In the terminology of MSR codes, this is called the data reconstruction property.  Apart from data reconstruction, an MSR code is required to be able to repair the $j$ node, by making use of the $d\beta$ symbols $\{S_{i \rightarrow j} \underline{c}_i: j \neq i,i \in [n]\}$ for some $(\beta \times \alpha)$ repair matrices $\{S_{i \rightarrow j}\}_{i \neq j}$. 

When $S_{i \rightarrow j}=S_j$ $,\forall i  \in [n] \setminus \{j\}$, we say repair matrices are independent of helper-node index $i$ (or constant repair matrix case).
Now lets consider the MSR code for any $d(\leq n-1)$. Let, $S^D_{i \rightarrow j}$ be $(\beta \times \alpha)$ the repair matrix where $S^D_{i \rightarrow j} \uc{i}$ is downloaded for the repair of the node $j$ when the helper nodes ($d$ nodes from which data is downloaded for the repair of node $j$) belong to the set $D$. 
As a result, the input to the replacement of failed node $j$ is the set:
\bean
\left\{   S^D_{i \rightarrow j} \uc{i}\mid    i \in D \right\},
\eean 
for helper nodes in the set $D \subseteq [n] \setminus \{j\}$ such that $|D|=d$. Here also, the code symbol 
\bean
\uc{j} = f_{D,j} \left( \left\{   S^D_{i \rightarrow j} \uc{i}\mid    i \in [n], \ i \neq j  \right\} \right)
\eean
where $f_{D,j}:\fq^{d\beta} \rightarrow \fq^{\alpha}$ is a deterministic linear function.
When $S^D_{i \rightarrow j}=S_{i \rightarrow j}$, we say that repair matrices are independent of identity of remaining helper nodes. Similar to $d=n-1$ case, the term {\em optimal access} MSR code for any $d$ means that the rows of $S^D_{i \rightarrow j}$ are picked from standard basis $\{e_1,\hdots,e_{\alpha}\}$ of $\fq^{\alpha}$.  We drop the superscript $D$ in  $S^D_{i \rightarrow j}$ when $d=n-1$.

\subsection{Optimal-Access MDS Codes} The present paper also includes results which are extended to an $[n,k]$ MDS code over the vector alphabet $\mathbb{F}_q^{\alpha}$, having the property that it can repair the failure of any node in a subset of $w (\leq n-1)$ nodes where failure of each particular node, can be carried out using linear operations, by uniformly downloading $\beta$ symbols from a collection of $d$ helper nodes. Linear repair of any node among the $w$ nodes is defined as in MSR code using repair matrices. It can be shown that even here, the minimum amount $d\beta$ of data download needed for repair of a single failed node, is given by $d\beta=d\frac{\alpha}{(d-k+1)}$ (minimum repair bandwidth). Our objective here as well, is on lower bounds on the sub-packetization level $\alpha$ of an MDS code that can carry out repair of any node in a subset of $w$ nodes, $1 \leq w \leq (n-1)$ where each node is repaired (linear repair) by optimal access with minimum repair bandwidth.   We prove a lower bound on $\alpha$ for the case of $d=(n-1)$.  This bound holds for any $w (\leq n-1)$ and is shown to be tight, again by comparing with recent code constructions in the literature. Also provided, are bounds for the case $d<(n-1)$. The $w=n$ case correspond to the MSR code case which is described before. The lower bound on $\alpha$ for $w=n$ case is derived separately as mentioned before and the bound is shown to be tight by comparing with recent code constructions when $(n-k) \nmid (n-1)$. Although the bound is tight when $(n-k) \nmid (n-1)$, the bound is valid for all the parameters (with $d=n-1$).

\subsection{Prior Work on MSR Codes }
Several constructions of MSR code can be found in the literature. In addition, there are constructions of systematic MDS codes in the literature where it is only the systematic nodes that can be recovered with minimal repair bandwidth, i.e., repaired by downloading $d\alpha / (d-k+1)$ symbols. We will refer to this latter class of codes as systematic MSR codes. A detailed survey on the MSR constructions and sub-packetization level bounds can be found in \cite{BalNikVaj}. The product matrix construction in \cite{RasShaKum_pm} for any $2k-2 \le d \le n-1$ is one of the first constructions of an MSR code. In \cite{PapDimCad}, the authors provide a high-rate MSR construction using Hadamard designs for any $(n,k=n-2,d=n-1)$ parameters. In \cite{TamWanBru}, high-rate systematic node repair MSR codes called Zigzag codes were constructed for $d=n-1$. These codes however had large field size and sub-packetization that is exponential in $k$. This construction was extended in \cite{WanTamBru_allerton} to enable the repair of parity nodes. The existence of MSR codes for any value of $(n,k,d)$  as $\alpha$ tends to infinity is shown in \cite{CadJafMalRamSuh}. 
The authors of  \cite{AgaSasKum}, introduced an optimal access systematic MSR construction for the case $d=n-1$ with $\alpha = r^{\frac{k}{r}}$. This was followed by \cite{SasAgaKum}, where the authors come up with optimal access MSR for the case $d=n-1$ with $\alpha = r^{\lceil \frac{n}{r}\rceil}$. The construction in \cite{SasAgaKum} was extended to any $d \le n-1$ in \cite{RawKoyVis_msr} with $\alpha = s^{\lceil \frac{n}{s} \rceil}$. The constructions in \cite{AgaSasKum, SasAgaKum, RawKoyVis_msr} are not explicit and need larger field size. In \cite{YeBar_1} explicit MSR constructions for any $(n, k, d)$ with field size $O(n)$ and $\alpha = s^n$ are proposed. Later in \cite{YeBarg_2017, LiTangTian, SasVajKum} there were three independent discoveries of the Coupled Layer (Clay) MSR code for parameters $(n, k, d=n-1)$ with sub-packetization $\alpha = r^{\lceil \frac{n}{r} \rceil}$. 

An open problem in the literature on regenerating codes is that of determining the smallest value of sub-packetization level $\alpha$ of an MSR code, given the parameters $\{(n,k,d=(n-1)\}$. This question is addressed in \cite{TamWanBru_access}, where a lower bound on $\alpha$ for MSR codes is presented by showing that $k \le \alpha {\alpha \choose \alpha / (n-k)}$. The authors of \cite{TamWanBru_access} have also shown that for the special case of optimal-access MSR codes, the sub-packetization level is lower bounded as
\bean
\alpha \geq r^{\frac{k-1}{r}}.
\eean
In \cite{GopTamCal} it is established that:
\bean
k \leq 2 \log_2(\alpha) (\lfloor \log_{\frac{r}{r-1}} (\alpha) \rfloor + 1),
\eean
while more recently, in \cite{HuangParamXian} the authors prove that:
\bean
k \leq 2 \log_r(\alpha) (\lfloor \log_{\frac{r}{r-1}} (\alpha) \rfloor + 1)
\eean
for MSR codes. Recently, in \cite{OmarGur} the authors prove that $ \alpha \ge e^{\frac{k-1}{4r}}$ for an MSR code.


\subsection{Our Contributions} 
\subsubsection{Improved Lower Bound on Sub-packetization Level for MSR codes}
We present an improved lower bound on the sub-packetization level,  $\alpha$ of an $(n, k, d=n-1)$ MSR code given by:
\bea
\alpha \ge e^{\frac{(k-1)(r-1)}{2r^2}}, \label{eq:subpkt_1}
\eea
by improving upon the lower bound $\alpha \ge e^{\frac{k-1}{4r}}$ presented by the authors in \cite{OmarGur}. The bound we present improves the constant $4$ to $\frac{2r}{(r-1)} \le 4$. We generalise the lemma's presented in \cite{OmarGur} by using the property that dimension of repair subspace intersections reduces $r$-fold by each added intersection. This property was first shown for the constant repair subspace case in \cite{TamWanBru}. We extend it for any repair subspace case in this paper in Lemma~\ref{lem:rsi}.\\

\subsubsection{Improved Lower Bound on Sub-packetization Level for Optimal Access MSR codes}
We also present an improved lower bound for subpacketization level $\alpha$ of an $(n, k, d=n-1)$ optimal-access MSR code given by:
\bea
\alpha \ge \min \{r^{\lceil\frac{n-1}{r}\rceil}, r^{k-1} \} \text{ where } r=n-k. \label{eq:subpkt_2}
\eea
This is an improvement over the lower bound presented in \cite{TamoWangBruck} given by $\alpha \ge r^{\frac{k-1}{r}}$.

\ben
\item The approach we follow to derive the bound for the optimal access MSR codes is along the lines of that adopted in \cite{TamoWangBruck} but with the difference that here we consider non-constant repair subspaces and consider all-node repair. We first derive the bound for the case $d=n-1$ as shown in \eqref{eq:subpkt_2} and extend it to general $d$. We then extend the bounds to optimal MDS code where, given the failed node belongs to a fixed set of $w$ nodes, optimal access repair is possible.
\item A tabular summary of the new lower bounds on sub-packetization-level $\alpha$ derived here for the optimal access case appears in Table~\ref{tab:results}.\\
\een

\begin{table}[h!] \caption{A tabular summary of the new lower bounds on sub-packetization-level $\alpha$ contained in the present paper.   In the table, the number of nodes repaired is a reference to the number of nodes repaired with minimum possible repair bandwidth $d\beta$ and moreover, in help-by-transfer fashion.  An * in the first column indicates that the bound is tight, i.e., that there is a matching construction in the literature.  With respect to the  first entry in the top row, we note that all MSR codes are MDS codes.  \label{tab:results} }
\begin{center}
	\scalebox{0.8}{
	\begin{TAB}(r,1cm,1.5cm)[5pt]{|c|c|c|c|c|c|c|}{|c|c|c|c|c|c|}
		 \shortstack{MSR \\ or MDS \\  Code ?} & $d$ & \shortstack{No. of Nodes \\ Repaired} & \shortstack{Assumption on \\ Repair Matrices \\ $S_{(i,j)}$ } & \shortstack{Lower Bound on $\alpha$ \\ and Constructions achieving our \\ lower bound on $\alpha$} & \shortstack{Reference \\ in this paper} & \shortstack{Previous \\ Known Bound} \\ 
         \shortstack{MSR* } & $n-1$ & $n$ & none & \shortstack{$\alpha \geq  \min \{r^{\lceil \frac{n-1}{r} \rceil},r^{k-1}\}$ \\ Constructions (when $r \nmid (n-1)$): \cite{YeBarg_2017,SasVajKum}} & Theorem \ref{thm:main} & $\alpha \geq r^{\frac{k-1}{r}}$ \\ 
           \shortstack{MSR*}  & $n-1$ & $n$ & \shortstack{independent of \\ helper-node index $i$} & \shortstack{$\alpha \geq  \min \{ r^{\lceil \frac{n}{r} \rceil},r^{k-1}\}$ \\ Constructions: \cite{YeBarg_2017,SasVajKum}} & Corollary \ref{sb_5} & $\alpha \geq r^{\frac{k}{r}}$ \\        
  		      \shortstack{MSR}  & any $d$ & $n$ & \shortstack{independent \\ of identity of \\ remaining \\ helper nodes } &  \shortstack{ $s=d-k+1$ \\ $\alpha \geq  \min \{ s^{\lceil \frac{n-1}{s} \rceil},s^{k-1} \}$} & Corollary \ref{sb_6} & none\\       
      MDS* & $n-1$ & $w$ ($\leq n-1$) & none & \shortstack{$\alpha \geq  \left\{ \begin{array}{rl} \min \{ r^{\lceil \frac{w}{r} \rceil},r^{k-1} \}, \  & w > (k-1)  \\
         r^{\lceil \frac{w}{r} \rceil}, \ & w \leq (k-1) . \end{array} \right. $ \\ Construction: \cite{LiTangTian}} & Corollary \ref{cor:node_subset} & \shortstack{for $w=k$ \\ $\alpha \geq r^{\frac{k-1}{r}}$}  \\MDS & any $d$ & $w$ ($\leq d$) & none &  \shortstack{ $s=d-k+1$ \\ $\alpha \geq  \left\{ \begin{array}{rl} \min \{ s^{\lceil \frac{w}{s} \rceil},s^{k-1}\}, \  & w > (k-1) \\
s^{\lceil \frac{w}{s}\rceil}, \ & w \leq (k-1) . \end{array} \right. $} & Corollary \ref{sb_4} & none \\ 
\end{TAB}}
\end{center}
\end{table}

Comparison of our lower bound on $\alpha$ for optimal-access MSR codes with existing code constructions:
\ben
\item When $r \nmid n-1$, our bound on $\alpha$ (Theorem \ref{thm:main}) for optimal access MSR code with $d=n-1$ becomes:
\bean
\alpha \geq  \min \{r^{\lceil \frac{n}{r} \rceil},r^{k-1}\} 
\eean
For $k \geq 5$, this reduces to $\alpha \geq  r^{\lceil \frac{n}{r} \rceil}$. 
The latter lower bound on $\alpha$ is achievable by the constructions in \cite{YeBarg_2017,SasVajKum, LiTangTian, SasAgaKum}. Hence our lower bound on $\alpha$ is tight.
Although our bound on $\alpha$ is shown to be tight only for $(n-1)$ not a multiple of $(n-k)$, the lower bound is valid for all parameters (with $d=n-1$) and when $r$ divides $(n-1)$.
\item Our bound on $\alpha$ (Corollary \ref{cor:node_subset} ) for MDS code with optimal access repair (repair with minimum repair bandwidth and help-by-transfer repair) for a failed node when it belongs to a fixed set of $w (\leq n-1)$ nodes is:
\bean
\alpha \geq  \left\{ \begin{array}{rl} \min \{ r^{\lceil \frac{w}{r} \rceil},r^{k-1} \}, \  & w > (k-1) \\
	r^{\lceil \frac{w}{r} \rceil}, \ & w \leq (k-1) . \end{array} \right. 
\eean
The above bound for $k \geq 5$ becomes $\alpha \geq r^{\lceil \frac{w}{r}\rceil}$.
This lower bound is achieved by the construction given in \cite{LiTangTian}. Hence our lower bound on $\alpha$ for the repair of $w$ nodes is tight.
\item The constructions in \cite{YeBarg_2017,SasVajKum, SasAgaKum} have repair matrices that are independent of the helper node index $i$ i.e., $S_{i \rightarrow j}=S_j$ and has sub-packetization $\alpha = r^{\lceil \frac{n}{r} \rceil}$ which achieves our lower bound on $\alpha$ (Corollary \ref{sb_5}) under the assumption that $S_{i \rightarrow j}=S_j$ for an optimal access MSR code with $d=n-1$. 
\item Our bound given in Corollary \ref{sb_6} is achieved by construction in \cite{RawKoyVis_msr} when $(d-k+1)$ does not divide $n-1$ under the assumption $S^D_{i \rightarrow j}=S_{i \rightarrow j}$.
\een
Hence our lower bound on $\alpha$ is tight for four cases.

\subsubsection{Structure Theorems on Repair Subspaces and Parity Check Matrix}
We deduce the parity check matrix support structure of an MDS code that allows for optimal-access repair of any single node within a subset of $w$ nodes for the case when $w \le k$ under the conditions that 
\ben
\item[(a)] it has optimal sub-packetization level and, 
\item[(b)] has constant repair subspaces.
\een
The support structure we characterise here is in fact the exact support structure of the MDS codes presented in \cite{LiTangTian} where repair of a single node from a subset of $w \le k$ nodes is possible through optimal-access.


%
%

\subsection{Outline}
In Section~\ref{sec:msr_bound} we present the improved sub-packetization level bound for a general MSR code. We follow this up with an improved bound on $\alpha$ for the optimal access MSR case in Section~\ref{sec:msr_oa_bound}. In Section~\ref{sec:msr_oa_bound} we also present lower bounds on $\alpha$ for MDS codes that can repair any node from the set of $w < n$ nodes optimally. In Section~\ref{sec:struct_thm} we present theorems that would determine the parity check matrix support structure of an optimal-access MSR code that achieves optimal sub-packetization level. In Section~\ref{sec:clay} we show how the resultant support structure is same as the support structure of MDS codes presented in \cite{LiTangTian} that can repair by optimal-access any single node from a subset of $w \le k$ nodes. We briefly explain how the support structure of Clay code constructed in  \cite{SasVajKum, LiTangTian, YeBarg_2017} is similar to the structure presented for the case when $w \le k$ in Section~\ref{sec:supp_msr}.

\subsection{Notation}
We adopt the following notation throughout the paper. 
\ben
\item Given a matrix $A$, we use $<A>$ to refer to the row space of the matrix $A$, 
\item Given a subspace $V$ and a matrix $A$, by $VA$ we will mean the subspace 
$\{ \underline{v} A \mid \underline{v} \in V\}$ obtained through transformation of $V$ by $A$.
\bit
\item Thus for example, $(\bigcap_i <S_i>) A$ will indicate the subspace obtained by transforming the intersection subspace $(\bigcap_i <S_i>)$ through right multiplication by $A$.
\eit
\item Let $\uz = (z_0, z_1, \cdots, z_{t-1}) \in \Zst$ then for $x \in \Zs$ and $y \in \Zt$:
\bean
\uz(x \rightarrow z_y) = (z_0, \cdots, z_{y-1}, x, z_{y+1}, \cdots, z_{t-1})
\eean
 
\een

\section{Improved Lower Bound on the Sub-packetization Level for MSR Codes \label{sec:msr_bound}}

In this section we will present an improved lower on the sub-packetization level of MSR codes. 
%
%
We begin with some helpful notation.  We will use the indices (two disjoint subsets of $[n]$):
\bean
\{ u_1,u_2,\cdots, u_k \}, \ \{p_1,p_2,\cdots,p_r\}
\eean
to denote the $n$ nodes in the network over which the code symbols are stored (Note that the code symbol \uc{i} is stored in node $i$.).  Let $\ell, \ 2 \leq \ell \leq (k-1)$, be an integer and set
\bean
U & = & \{ u_1,u_2,\cdots, u_{\ell} \}, \\
V & = & \{ u_{\ell+1},u_{z+2},\cdots, u_k \}, \\
P & = & \{p_1,p_2,\cdots,p_r\} .
\eean
Note that our choice of $\ell$ ensures that neither $U$ nor $V$ is empty.  We will now present the Lemma~\ref{lem:rsi} on the dimension of intersection of repair subspaces. This property will be used in coming up with improved lower bounds for general MSR codes and optimal-access MSR codes later. The proof of this lemma uses the ideas presented in \cite{TamoWangBruck} where they present similar bound for dimension of intersection of repair subspaces but for the specific case of constant repair subspaces.
%

\begin{lem} \label{lem:rsi} \textbf{(Repair Subspace Intersection)}: 
    $\dim \left( \bigcap_{u \in U} <S_{p \rightarrow u} > \right)$ is the same for all $p \in P$ and
	\bea \label{eq:lemma_2}
	\sum_{p' \in P} \dim \left( \bigcap_{u \in U} < S_{p' \rightarrow u} > \right)  \ \leq \ 
	\dim \left( \bigcap_{u \in U \setminus \{u_{\ell}\}} <S_{p \rightarrow u}> \right), 
	\eea
	where $p$ is an arbitrary node in  $P$. 
\end{lem}	
\begin{proof} \ 
	
	\paragraph{Invariance of $\ell$-fold Intersection of Repair Subspaces Contributed by a Parity Node}  Let us consider the nodes in $U \cup V$ as systematic nodes and nodes in $P$ as parity nodes. Note that the sets $U,V,P$ are pairwise disjoint and are arbitrary subsets of $[n]$, under the size restrictions $2 \leq |U|=\ell \leq k-1,|P|=r$ and $U \cup V \cup P = [n]$. First we prove that $\dim \left(\bigcap_{u \in U} <S_{p \rightarrow u}> \right)$ is the same for all $p \in P$. Note that $< S_{p \rightarrow u} >$ is the row space of the repair matrix carrying repair information from helper (parity) node $p$ to the replacement of the failed node $u$.  Thus we are seeking to prove that the $\ell$-fold intersection of the $\ell$ subspaces $\{ <S_{p \rightarrow u}> \}_{u \in U}$ obtained by varying the failed node $u \in U$ is the same, regardless of the parity node $p \in P$ from which the helper data originates.   
	
	To show this, consider a generator matrix $G$ for the code $\mathcal{C}$ in which the nodes of $P$ are the parity nodes and the nodes in $U \cup V \ = \ \{u_1,u_2,\cdots,u_k\}$ are the systematic nodes.  Then $G$ will take on the form: 	
	\bea
	G = \left[
	\scalemath{1}{
		\begin{array}{cccc}
			I_{\alpha} & 0 & \hdots & 0 \\
			0 & I_{\alpha} & \hdots & 0 \\
			\vdots & \vdots & \vdots & \vdots \\
			0 & 0 & \hdots & I_{\alpha} \\
			A_{p_1,u_1} & A_{p_1,u_2} & \hdots & A_{p_1,u_k} \\
			A_{p_2,u_1} & A_{p_2,u_2} & \hdots & A_{p_2,u_k} \\
			\vdots & \vdots & \vdots & \vdots \\
			A_{p_r,u_1} & A_{p_r,u_2} & \hdots & A_{p_r,u_k} \\
		\end{array} 
	}
	\right]. \ \ \label{Gform_3}
	\eea
	Wolog the generator matrix assumes an ordering of the nodes in which the first $k$ nodes in $G$ (node $i$ in $G$ correspond to rows $[(i-1)\alpha+1,i\alpha]$) correspond respectively to $\{u_i \mid 1 \leq i \leq k\}$ and the remaining $r$ nodes in $G$ to $\{p_1,\cdots,p_r\}$ in the same order. 
	A codeword is formed by $G\underline{m}$ where $\underline{m}$ is the vector of $k \alpha$ message symbols over \fq\ encoded by the MSR code and the index of code symbols in the codeword is according to the nodes in $U,V,P$ i.e., the codeword will be of the form $[\uc{u_1}^T,\hdots,\uc{u_k}^T,\uc{p_1}^T,\hdots,\uc{p_r}^T]^T = G \underline{m}$.
	
	\begin{center}
		\begin{figure}
			\begin{center}
				\includegraphics[width=2.5in]{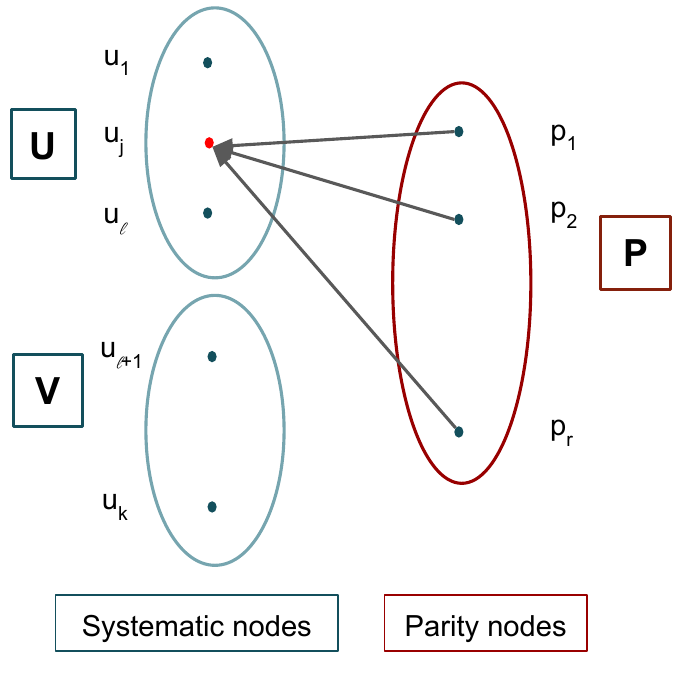}
			\end{center} 
			\caption{The general setting considered here where helper data flows from the parity-nodes $\{p_i\}_{i=1}^r$ forming set $P$ to a failed node $u_j \in U$. \label{fig:setting}}
		\end{figure}
	\end{center}

	By the interference-alignment conditions~\cite{ShahRashKumKann_ia}, \cite{TamoWangBruck} applied to the repair of a systematic node $u_j \in U$, we obtain 	(see Lemma \ref{lem:rank_cond_p_node} in the appendix, for a more complete discussion on interference alignment equations given below): 
	\bea
	<S_{p \rightarrow u_j} A_{p,u_{\ell+1}}>  & = &  <S_{p' \rightarrow u_j} A_{p',u_{\ell+1}}>, \ \text{ for every pair } p, p' \in P. \label{eq:IA_P_1}
	\eea	 
	
	Equation \eqref{eq:IA_P_1} and Lemma \ref{Row_Spaces} implies (as $A_{i,j}$ are invertible for all $i,j$):
	\bea
	\left ( \bigcap_{j=1}^{\ell} <S_{p \rightarrow u_j}> \right )  A_{p,u_{\ell+1}} = \left ( \bigcap_{j=1}^{\ell} <S_{p' \rightarrow u_j}> \right)  A_{p',u_{\ell+1}},\text{ for every pair } p, p' \in P . \label{eq:l_fold}
	\eea
	
	It follows then from the non-singularity of the matrices $A_{i,j}$ and equation \eqref{eq:l_fold},  that $\text{dim}(\bigcap_{u \in U} <S_{(p,u)}>)$ is the same for all $p \in P$.  It remains to prove the main inequality \eqref{eq:lemma_2}.    		
	
	\paragraph{$(\ell-1)$-fold Intersection of Repair Subspaces} 
	We proceed similarly in the case of an $(\ell-1)$-fold intersection, replacing $\ell$ by $\ell-1$ in \eqref{eq:l_fold}.
	We will then obtain: 
	\bea
	\left ( \bigcap_{j=1}^{\ell-1} <S_{p \rightarrow u_j}> \right )  A_{p,u_{\ell+1}} = \left ( \bigcap_{j=1}^{\ell-1} <S_{p' \rightarrow u_j}> \right)  A_{p',u_{\ell+1}},\text{ for every pair } p, p' \in P . \label{eq:l_minus_one}
	\eea
	\paragraph{Relating $\ell$-fold and $(\ell-1)$-fold intersections}	Next consider the repair of the node $u_{\ell}$. 
	It follows from \eqref{eq:l_minus_one} that for any $p', p \in P$:
	\bea \nonumber 
	\left ( \bigcap_{j=1}^{\ell} <S_{p \rightarrow u_j}> \right ) A_{p,u_{\ell}} 	& = & 
	<S_{p \rightarrow u_{\ell}} A_{p,u_{\ell}}> \bigcap \left ( \left ( \bigcap_{j=1}^{\ell-1} <S_{p \rightarrow u_j}> \right )  A_{p,u_{\ell}} \right )\\ 
	& \subseteq & \left ( \bigcap_{j=1}^{\ell-1} <S_{p' \rightarrow u_j}> \right )  A_{p',u_{\ell}}.  			\label{eq:l_to_lminus1}
	\eea
	As a consequence of \eqref{eq:l_to_lminus1} it follows that:
	\bea
	\label{eq:subset} \bigoplus_{i = 1}^r \left ( \bigcap_{u \in U} <S_{p_i \rightarrow u}> \right ) A_{p_i,u_{\ell}}
	& \subseteq & \left (  \bigcap_{j=1}^{\ell-1} <S_{p \rightarrow u_j}> \right ) A_{p,u_{\ell}} \text{ for any } p \in P.
	\eea			
	From the full-rank condition of node repair, (see Lemma \ref{lem:rank_cond_p_node}), we must have that 			
	\bea
	\text{rank} \left( \left[
	\scalemath{1}{
		\begin{array}{c}
			S_{p_1 \rightarrow u_{\ell}} A_{p_1,u_{\ell}} \\
			S_{p_2 \rightarrow u_{\ell}} A_{p_2,u_{\ell}} \\
			S_{p_3 \rightarrow u_{\ell}} A_{p_3,u_{\ell}} \\
			\vdots \\
			S_{p_r \rightarrow u_{\ell}} A_{p_r,u_{\ell}} \\
		\end{array} 
	}
	\right]\right) & = & \alpha. \ \ \label{Cond_1}
	\eea
	It follows as a consequence, that 
	\bea
	\bigoplus_{i = 1}^r < S_{p_i \rightarrow u_{\ell}} A_{p_i,u_{\ell}} > & = & \mathbb{F}_q^{\alpha},  \label{eq:stratify} 
	\eea 
	and hence, for every $p \in \ P$, we must have that 
	\bea
	< S_{p \rightarrow u_{\ell}} A_{p,u_{\ell}} >  \bigcap \ 		\bigoplus_{p' \in [P], p' \neq p  }  <S_{p' \rightarrow u_{\ell}} A_{p',u_{\ell}} > & = & \{ \underline{0} \} .  \label{eq:sum_disjoint} 
	\eea 			
	
	Since the $A_{i,j}$'s are nonsingular, from equations \eqref{eq:subset} and \eqref{eq:sum_disjoint} we can conclude that:
	\bea
	\sum_{i = 1}^r \dim \left ( \bigcap_{u \in U} <S_{ p_i \rightarrow u}> \right ) \ \leq \ 
	\dim \left ( \bigcap_{u \in U \setminus \{u_{\ell}\}} <S_{p \rightarrow u}> \right ). \label{eq:lemma22}
	\eea
	for any $p \in P$, which is precisely the desired equation \eqref{eq:lemma_2}. 
\end{proof} 

\begin{cor} \label{cor:rsi}
	\bea \label{eq:rsi_cor}
	\dim \left( \bigcap_{u \in U} < S_{p \rightarrow u} > \right)  \ \leq \ 
	\frac{1}{r}\dim \left( \bigcap_{u \in U \setminus \{u_{\ell}\}} <S_{p' \rightarrow u}> \right), 
	\eea	
	where $p, p'$ are arbitrary nodes in  $P$.   
\end{cor}
\bprf
The proof follows from Lemma~\ref{lem:rsi} as $	\dim \left( \bigcap_{u \in U} < S_{p \rightarrow u} > \right) = 	\dim \left( \bigcap_{u \in U} < S_{p' \rightarrow u} > \right)$ for any $p, p' \in P$.
\eprf

\begin{thm} \label{thm:main_gen} \textbf{(Subpacketization Bound)}: 
	Let $\mathcal{C}$ be a linear MSR code having parameter set 
	\bean
	\left\{ (n,k,d),\ (\alpha,\beta),\ B, \ \fq) \right\},
	\eean
	with $d=(n-1)$ and linear repair for all $n$ nodes.  Then we must have:
	\bea
	\alpha \geq  e^{\frac{(r-1)(k-1)}{2r^2}} \label{eq:Main_Bound_Gen}.
	\eea
\end{thm}
Before presenting the proof of this Theorem, we introduce some notation and lemmas presented in \cite{OmarGur} that will help us prove the improved bound. We start by restating the definition of MSR subspace family introduced in \cite{OmarGur} and restate the Lemma~\ref{lem:ssfamily_exists} presented in \cite{OmarGur} that states that for every MSR code there exists an MSR subspace family.

\begin{defn}[MSR Subspace Family] The set of sub spaces $\{H_1, H_2, \cdots, H_t\}$ is called an $(\alpha, r)_{\fq}$ subspace family if dim$(H_i) = \frac{\alpha}{r}$ for all $i \in [t]$ and if there exist invertible linear maps $\phi_{i, j}$ on $\fq^{\alpha}$ for $i \in [t]$ and $j \in [r-1]$ such that:
\bean
H_i \oplus \oplus_{j=1}^{r-1} \phi_{i, j}(H_i) &=& \fq^{\alpha} \text{ for all } i \in [t],\\
\phi_{i', j} (H_i) &=& H_i \text{ for all } i \ne i' \in [t], j \in [r-1].
\eean
\end{defn}

\blem[Theorem 2 in \cite{TamWanBru_access}] \label{lem:constss_tamo}
If there exists an MSR code with parameters $(n, k, d=n-1, \alpha)$ and repair matrices $S_{i \rightarrow u_j}$ for $j \in [k]$, $i \in [n]\setminus \{u_j\}$ then there exists an MSR code that has constant repair matrices ( $S_{u_j} = S_{p_r \rightarrow u_j}$ for $j \in [k-1]$) with parameters $(n-1, k-1, d=n-2, \alpha )$.
\elem

\blem[Proposition 2 in \cite{OmarGur}] \label{lem:ssfamily_exists} 
Suppose there exists an MSR code with parameters $(n, k, d=n-1, \alpha)$ over $\fq$ with repair matrices $S_{i \rightarrow u_j}$ for $j \in [k], i \in [n] \setminus \{u_j\}$. Then there exists an $(\alpha, r)_{\fq}$-MSR subspace family with $(k-1)$ subspaces $H_1, \cdots, H_{k-1}$ given by $H_i = <S_{p_r \rightarrow u_i}>$.
\elem


\bcor\label{cor:msrfam_int}
Let $\{H_1, \cdots, H_{k-1}\}$ be an $(\alpha, r)_{\fq}$ MSR subspace family defined from an $(n, k, d=n-1, \alpha)$ MSR code as in Lemma~\ref{lem:ssfamily_exists}. Then:
\bean
\text{dim}( \cap_{i \in S} H_i) \le \frac{\alpha}{r^{|S|}} \text{ where } S \subseteq [k-1].
\eean
\ecor
\bprf
By Lemma~\ref{lem:ssfamily_exists}, $H_i = < S_{p_r \rightarrow u_i}>$ for all $i \in [k-1]$. Therefore by repeated application of Corollary~\ref{cor:rsi} it follows that:
\bean
dim(\cap_{i \in S} H_i) \le \frac{\alpha}{r^{|S|}},
\eean
as dim$(H_i) = \frac{\alpha}{r}$ for any $i \in [k-1]$.
\eprf

In order to prove the lower bound on sub-packetization level the authors of \cite{OmarGur} look at the dimension of a vector space of linear functions that satisfy certain properties and come up with upper bound for it. In this paper we improve the upper bound by using the upper bound on the dimension of intersection of repair subspaces derived in Corollary~\ref{cor:msrfam_int}. We will now start by introducing the notation for subspace of linear functions whose dimension we would like to upper bound.

Let $\call (\fq^{\alpha}, \fq^{\alpha})$ be the vector space of all linear maps from $\fq^{\alpha}$ to $\fq^{\alpha}$, then the subspace:
\bean
\calf (A_1 \rightarrow B_1, A_2 \rightarrow B_2, \cdots, A_n \rightarrow B_n) &=& \{ \psi \in \call(\fq^{\alpha}, \fq^{\alpha}) \mid \psi(A_i) \subseteq B_i, \ \forall i \in [n]\}\\
I (A_1 \rightarrow B_1, A_2 \rightarrow B_2, \cdots, A_n \rightarrow B_n) &=& dim(\calf (A_1 \rightarrow B_1, A_2 \rightarrow B_2, \cdots, A_n \rightarrow B_n)),
\eean
and $I(A_1, \cdots, A_n) = I (A_1 \rightarrow A_1, A_2 \rightarrow A_2, \cdots, A_n \rightarrow A_n)$, $I(A_1, \cdots, A_{n-1}, A_n \rightarrow B_n) = I (A_1 \rightarrow A_1,  \cdots, A_{n-1} \rightarrow A_{n-1}, A_n \rightarrow B_n)$. The lower bound on the sub-packetization level $\alpha$ is presented by coming up with an upper bound on $I(H_1, \cdots, H_{k-1})$ where $\{H_1, \cdots, H_{k-1} \}$ is an $(r, \alpha)_{\fq}$ subspace family in Lemma~\ref{lem:subspace_dim}. We will first present an upper bound on sum of dimension of sub-spaces in Lemma~\ref{lem:sumdim}. This bound generalizes Lemma~4 presented in \cite{OmarGur} where the subspaces $A_1, \cdots, A_n$ were considered to satisfy $\cap_{i=1}^n A_i = \phi$. The generalization of this bound will help in coming up with an improved lower bound on the sub-packetization level presented in Theorem~\ref{thm:main_gen}.

%

\blem \label{lem:sumdim}
Let $A_1, \cdots, A_n$ be $n$ sub spaces of $\fq^{\alpha}$ then:
\bean
\label{eq:sumdim}
\sum\limits_{i=1}^n dim(A_i) \le (n-1) dim(\sum_{i=1}^n A_i) + dim (\cap_{i=1}^n A_i)
\eean
\elem
\bprf
For $n= 2$, this statement follows as 
\bean
dim(A_1) + dim(A_2) = dim(A_1 + A_2) + dim(A_1 \cap A_2).
\eean
We will show by induction that it is true for any $n \ge 2$. Let us assume that equation~\eqref{eq:sumdim} is true for $n=k$, then we will show that it is true for $n=k+1$.
\bean
dim(A_1) + dim(A_2) + \sum\limits_{i=3}^{k+1} dim(A_i) & = &  dim(A_1 + A_2) + \sum\limits_{i=3}^{k+1} dim(A_i) + dim(A_1 \cap A_2)\\
& \le & dim(A_1+A_2) + (k-1)dim(\sum\limits_{i=3}^{k+1} A_i + (A_1 \cap A_2)) + dim(\cap_{i=1}^{k+1} A_i)\\
& \le & k dim(\sum\limits_{i=1}^{k+1} A_i) +  dim(\cap_{i=1}^{k+1} A_i).
\eean
\eprf

We will now use Lemma~\ref{lem:sumdim} to come up with a recursive upper bound on $I(H_1, \cdots, H_t)$ in the following Lemma.

\blem \label{lem:subspace_rec}
Let $H_1, H_2, \cdots , H_k$ be an $(\alpha, r)_{\fq}$ MSR subspace family, then for any $1 \le t \le k$ and $1 \le p \le k-t+1$:
\bean
r I(H_1, H_2, \cdots, H_t ) &\le&  (r-1)I(H_1, H_2, \cdots, H_{t-1} ) + I(H_1, H_2, \cdots, H_{t-1}, \fq^{\alpha} \rightarrow H_t)\\
r I(H_1, H_2, \cdots, H_{t-1}, \fq^{\alpha} \rightarrow \cap_{i=1}^{p} H_{t-1+i}) &\le&  (r-1)I(H_1, H_2, \cdots, H_{t-2}, \fq^{\alpha} \rightarrow \cap_{i=1}^{p} H_{t-1+i} )  \\ && \ \ \ + I(H_1, H_2, \cdots, H_{t-2}, \fq^{\alpha} \rightarrow \cap_{i=0}^{p} H_{t-1+i}).
\eean
\elem
\bprf Let $\phi_{t, 0}$ be an invertible linear map such that $\phi_{t, 0}(H_t) = H_t$. We prove this lemma by using the bijection between linear maps in $\calf_j = \calf (H_1 \rightarrow H_1, \cdots, H_{t-1} \rightarrow H_{t-1}, \phi_{t, j} (H_t) \rightarrow H_t)$ and linear maps in $\calf' = \calf (H_1 \rightarrow H_1, \cdots, H_{t-1} \rightarrow H_{t-1}, H_t \rightarrow H_t)$. For a linear map in $\psi \in \calf_j$ it can be verified that  $\psi \circ \phi_{t, j} \in \calf'$ and for any $\psi' \in \calf'$, $ \psi' \circ \phi_{t, j}^{-1} \in \calf_j$. Therefore $I(H_1, H_2, \cdots, H_t) = I(H_1, \cdots, H_{t-1}, \phi_{t, j}(H_t) \rightarrow H_t)$.

\bea
\nonumber r I(H_1, H_2, \cdots, H_t )  &=& \sum\limits_{j=0}^{r-1} I \left(H_1, H_2, \cdots, \phi_{t, j}(H_t) \rightarrow H_t\right)\\
\nonumber &\le& (r-1) dim\left(\oplus_{j=0}^{r-1} \calf(H_1, \cdots, H_{t-1}, \phi_{t, j}(H_t) \rightarrow H_t)\right)   + \\
\nonumber && \ \ \ dim\left(\cap_{j=0}^{r-1} \calf(H_1, \cdots, H_{t-1}, \phi_{t, j}(H_t) \rightarrow H_t)\right) \text{ from Lemma~\ref{lem:sumdim}}\\
\label{eq:ss_iter} &\le& (r-1) I\left(H_1, \cdots, H_{t-1}\right) + I\left(H_1, \cdots, H_{t-1}, \fq^{\alpha} \rightarrow H_t\right).
\eea
The equation~\eqref{eq:ss_iter} above follows as $\oplus_{j=0}^{r-1} \phi_{t, j}(H_t) = \fq^{\alpha}$. Similarly,
\bean
\nonumber r I(H_1, \cdots, H_{t-2}, H_{t-1}, \fq^{\alpha} \rightarrow \cap_{i=1}^{p} H_{t-1+i} )  &=& \sum\limits_{j=0}^{r-1} I \left(H_1, \cdots, H_{t-2}, \phi_{t-1, j}(H_{t-1}) \rightarrow H_{t-1}, \fq^{\alpha} \rightarrow \cap_{i=1}^{p} H_{t-1+i} \right)\\
&\le& (r-1) I\left(H_1, \cdots, H_{t-2}, \fq^{\alpha} \rightarrow \cap_{i=1}^{p} H_{t-1+i}\right) \\
&& \ \ \ \ + I\left(H_1, \cdots, H_{t-2}, \fq^{\alpha} \rightarrow H_{t-1} \cap \cap_{i=1}^{p} H_{t-1+i} \right).
\eean
\eprf
\blem \label{lem:subspace_dim} Let $H_1, H_2, \cdots , H_k$ be an $(\alpha, r)_{\fq}$ MSR subspace family, then for any $1 \le t \le k$ and $1 \le p \le k-t$:
\bea
\label{eq:subspace_dim}I(H_1, \cdots, H_{t}) &\le& \left(1 - \frac{1}{r} + \frac{1}{r^2}\right)^{t} \alpha^2\\
\label{eq:subspace_dim2}I(H_1, \cdots, H_{t}, \fq^{\alpha} \rightarrow \cap_{i=1}^{p} H_{t+i} ) &\le& \frac{1}{r^{p}}  \left(1 - \frac{1}{r} + \frac{1}{r^2}\right)^{t} \alpha^2 
\eea
\elem
\bprf
We will prove this by induction. We will first show that equations \eqref{eq:subspace_dim} and \eqref{eq:subspace_dim2} hold for $t=1$. From Lemma~\ref{lem:subspace_rec} it follows that:
\bean
rI(H_1) &\le& (r-1) \alpha^2 + I(\fq^{\alpha} \rightarrow H_1)\\
I(H_1) &\le& \left(1 - \frac{1}{r} + \frac{1}{r^2}\right) \alpha^2,
\eean
as the space of all linear maps in $\fq^{\alpha}$ has dimension $\alpha^2$ and $I(\fq^{\alpha} \rightarrow H_1) \le \frac{\alpha^2}{r}$. Similarly:
\bean
I\left(H_1, \fq^{\alpha} \rightarrow \cap_{i=1}^p H_{i+1}\right) &\le& \frac{r-1}{r} I(\fq^{\alpha} \rightarrow  \cap_{i=1}^p H_{i+1}) + \frac{1}{r} I(\fq^{\alpha} \rightarrow \cap_{i=0}^{p} H_{i+1}) \text{ due to Lemma~\ref{lem:subspace_rec}}\\
&\le& \frac{r-1}{r} \left( \frac{\alpha^2}{r^p} \right) + \frac{1}{r} \left( \frac{\alpha^2}{r^{p+1}} \right) = \frac{1}{r^p}\left(1 - \frac{1}{r} + \frac{1}{r^2}\right)\alpha^2,
\eean
as dim$(\cap_{i=1}^{p} H_{i+1}) \le \frac{\alpha}{r^p}$ it follows that $I(\fq^{\alpha} \rightarrow \cap_{i=1}^{p} H_{i+1}) \le \frac{\alpha^2}{r^p}$. Let us assume that equation \eqref{eq:subspace_dim}, \eqref{eq:subspace_dim2} hold for $t=m-1$, we will first show that \eqref{eq:subspace_dim} holds for $t=m$ and then show that \eqref{eq:subspace_dim2} holds for $t=m$. By Lemma~\ref{lem:subspace_rec} we get:
\bean
I(H_1, \cdots, H_{m}) &\le& \left(1-\frac{1}{r}\right) I(H_1, \cdots, H_{m-1}) + \frac{1}{r} I(H_1, \cdots, H_{m-1}, \fq^{\alpha} \rightarrow H_{m})\\
& \le & \left(1-\frac{1}{r}\right) \left(1 - \frac{1}{r}+\frac{1}{r^2}\right)^{m-1} \alpha^2 + \frac{1}{r^2} \left(1 - \frac{1}{r}+\frac{1}{r^2}\right)^{m-1} \alpha^2\\
& = & \left(1 - \frac{1}{r} + \frac{1}{r^2}\right)^{m} \alpha^2.
\eean
Similarly:
\bean
I\left(H_1, \cdots, H_{m}, \fq^{\alpha} \rightarrow \cap_{i=1}^p H_{m+i}\right) &\le& \left(1-\frac{1}{r}\right) I\left(H_1, \cdots, H_{m-1}, \fq^{\alpha} \rightarrow \cap_{i=1}^p H_{m+i}\right) \\ && \ \ \ + \frac{1}{r} I\left(H_1, \cdots, H_{m-1}, \fq^{\alpha} \rightarrow \cap_{i=0}^p H_{m+i}\right)\\
& \le & \frac{1}{r^p}\left(1-\frac{1}{r}\right) \left(1 - \frac{1}{r}+\frac{1}{r^2}\right)^{m-1} \alpha^2 + \frac{1}{r^{p+2}} \left(1 - \frac{1}{r}+\frac{1}{r^2}\right)^{m-1} \alpha^2\\
& = & \frac{1}{r^p} \left(1 - \frac{1}{r} + \frac{1}{r^2}\right)^{m} \alpha^2.
\eean
\eprf
\ \\
\emph{Proof of Theorem~\ref{thm:main_gen}}
	From Lemma~\ref{lem:subspace_dim} it follows that if there is an $(\alpha, r)_{\fq}$ MSR subspace family $\{H_1, \cdots, H_{k-1}\}$ of size $(k-1)$ then:
	\bean
	I(H_1, \cdots, H_{k-1}) &\le& \left(1 - \frac{1}{r} + \frac{1}{r^2}\right)^{k-1} \alpha^2
	\eean
	Identity map is an element in $\calf(H_1 \rightarrow H_1, \cdots, H_{k-1} \rightarrow H_{k-1})$. Therefore:
	\bean
	1 &\le& I(H_1, \cdots, H_{k-1}) \le \left(1 - \frac{1}{r} + \frac{1}{r^2}\right)^{k-1} \alpha^2\\
	\alpha &\ge& e^{\frac{(k-1)}{2 \log\left(\frac{r^2}{r^2-r+1}\right)}} \ge e^{\frac{(k-1)(r-1)}{2r^2}} \text{ as } log(1+x) \ge \frac{x}{1+x}.
	\eean
\eprf
\section{Improved Lower Bound on Subpacketization Level For Optimal Access MSR Codes \label{sec:msr_oa_bound}}
In this section we first start by presenting an improved lower bound for the sub-packetization level of optimal-access MSR codes. We then present lower bounds for the case of MDS codes where single node repair of $w \le k$ nodes can be done optimally. We also present lower bounds for the case $d < n-1$.

\begin{thm} \label{thm:main} \textbf{ (Subpacketization Bound for Optimal Access case)}: 
	Let $\mathcal{C}$ be a linear optimal access MSR code having parameter set 
	\bean
	\left\{ (n,k,d),\ (\alpha,\beta),\ B, \ \fq) \right\},
	\eean
	with $d=(n-1)$ and linear repair for all $n$ nodes.  Then we must have:
	\bean
	\alpha \geq  \min \{r^{\lceil \frac{n-1}{r} \rceil},r^{k-1}\} \label{eq:Main_Bound}.
	\eean
\end{thm}

\vspace*{0.3in}

\begin{proof} 
	
	\begin{enumerate}
		\item {\em Invariance of Repair Matrices to Choice of Generator Matrix} We first observe that the repair matrices can be kept constant, even if the generator matrix of the code changes.  This is because the repair matrices only depend upon relationships that hold among code symbols of any codeword in the code and are independent of the particular generator matrix used in encoding.  In particular, the repair matrices are insensitive to the characterization of a particular node as being either a systematic or parity-check node.
		\item {\em Implications for the Dimension of the Repair Subspace} 
		From Lemma~\ref{lem:rsi}, we have that 
		\bea 
		\sum_{i=1}^{r} \dim \left ( \bigcap_{u \in U} <S_{p_i \rightarrow u}> \right )  \ \leq \ 
		\dim \left ( \bigcap_{u \in U \setminus \{u_{\ell}\}} <S_{p \rightarrow u}> \right ), 
		\eea
		and moreover that $\dim \left (\bigcap_{u \in U} <S_{(p,u)}> \right )$ is the same for all $p \in P$. 
		It follows that 
		\bea 
		r \times \dim \left (\bigcap_{u \in U} <S_{p \rightarrow u}> \right ) \leq \dim \left ( \bigcap_{u \in U \setminus \{u_{\ell}\}} <S_{p \rightarrow u}> \right ) \label{Cond_3},
		\eea
		i.e., 
		\bea \nonumber 
		\dim \left ( \bigcap_{u \in U} <S_{p \rightarrow u}> \right ) & \leq & \frac{\dim \left (\bigcap_{u \in U\setminus \{u_{\ell}\}} <S_{p \rightarrow u}> \right )}{r}  \\ \nonumber 
		&  \leq & \frac{\dim \left ( \bigcap_{u \in U\setminus \{u_{\ell},u_{\ell-1}\}} <S_{p \rightarrow u}> \right )}{r^2}  \\ \nonumber 
		& \leq & \frac{\dim \left (<S_{p \rightarrow u_1}> \right )}{r^{\ell-1}} \\ 
		& = & \frac{\alpha}{r^{\ell}} \ = \  \frac{\alpha}{r^{|U|}}.
		\label{eq:Cond_4}
		\eea	   
		Lemma \ref{lem:rsi} and its proof holds true for any set $U \subseteq [n]$ of size $2 \leq |U| \le (k-1)$.    As a result, equation \eqref{eq:Cond_4}, also holds for any set $U \subseteq [n]$ of size $2 \leq |U| \le (k-1)$.   
		
		We would like to extend the above inequality to hold even for the case when $U$ is of size $k \leq |U| \leq (n-1)$.  We get around the restriction on $U$ as follows.   It will be convenient in the argument, to assume that $U$ does not contain the $n$th node, i.e., $U \subseteq [n-1]$ and $ n \in P$.   Let us next suppose that  $\alpha < r^{k-1}$ and that $U$ is of size $(k-1)$. We would then have: 
		\bea \label{eq:intersection}
		\dim \left ( \bigcap_{u \in U} <S_{n \rightarrow u}> \right ) & \leq  & \frac{\alpha}{r^{k-1}} \ < \ 1,
		\eea
		which is possible iff
		\bean
		\dim \left(  \bigcap_{u \in U} <S_{n \rightarrow u}> \right) & = & 0.
		\eean
		But this would imply that 
		\bean
		\dim \left(  \bigcap_{u \in F} <S_{n \rightarrow u}> \right) & = & 0.
		\eean
		for any subset $F\subseteq [n-1]$ of nodes of size $|F|$ satisfying $(k-1) \leq |F| \leq (n-1)$.  We are therefore justified in extending the inequality in \eqref{eq:Cond_4} to the case when $U$ is replaced by a subset $F$ whose size now ranges from $2$ to $(n-1)$, i.e., we are justified in writing: 
		\bea \label{eq:size_bd} 
		\dim \left ( \bigcap_{u \in F} <S_{n \rightarrow u}> \right ) &  \leq  & \frac{\alpha}{r^{|F|}} 
		\eea	   
		for any $F \subseteq [n-1]$, of size $2 \leq |F| \leq (n-1)$.  A consequence of the inequality \eqref{eq:size_bd} is that 
		\bean
		\dim \left ( \bigcap_{u \in F} <S_{n \rightarrow u}> \right ) & \geq 1 ,
		\eean
		implies that $|F| \leq \lfloor \log_{r} (\alpha) \rfloor $.    In other words, a given non-zero vector can belong to at most $\lfloor \log_{r} (\alpha) \rfloor$ repair subspaces among the repair subspaces: $<S_{n \rightarrow 1}>,\hdots,<S_{n \rightarrow n-1}>$.
		
		\item {\em Counting in a Bipartite Graph} The remainder of the proof then follows the steps outlined in Tamo et. al \cite{TamoWangBruck}.   We form a bipartite graph with $\{\mathbf{e}_1,...,\mathbf{e}_{\alpha}\}$ (standard basis) as left nodes and $S_{n \rightarrow 1},...,S_{n \rightarrow n-1}$ as right nodes as shown in Fig.~\ref{fig:biparite}. We declare that edge $(\mathbf{e}_i,S_{n \rightarrow j})$ belongs to the edge set of this bipartite graph iff $\mathbf{e}_i \in <S_{n \rightarrow j}>$. Now since the MSR code is an optimal access code, the rows of each repair matrix $S_{n \rightarrow j}$ must all be drawn from the set $\{\mathbf{e}_1,...,\mathbf{e}_{\alpha}\}$. 
		
		\begin{center}
			\begin{figure}
				\begin{center}
					\includegraphics[width=3in]{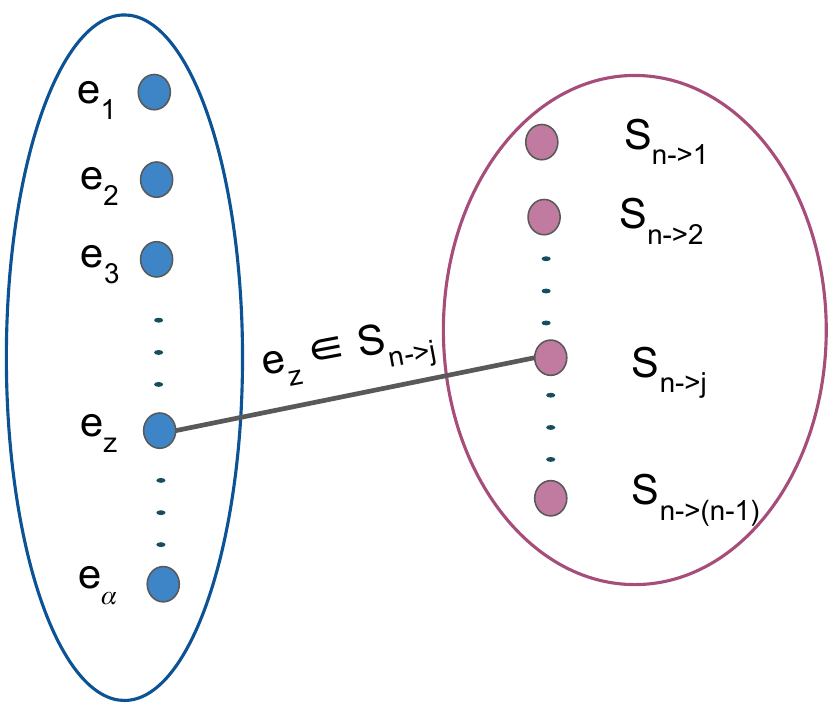}
				\end{center} 
				\caption{The above figure shows the bipartite graph appearing in the counting argument used to provie Theorem~\ref{thm:main}.  Each node on the left corresponds to an element of the standard basis $\{\mathbf{e}_1,...,\mathbf{e}_{\alpha}\}$.  The nodes to the right are associated to the repair matrices $S_{n \rightarrow 1},...,S_{n \rightarrow n-1}$.}
			\end{figure} \label{fig:biparite}
		\end{center}
		Counting the number of edges of this bipartite graph in terms of node degrees on the left and the right, we obtain:
		\bean
		\alpha \lfloor \log_r(\alpha) \rfloor \geq (n-1)\frac{\alpha}{r},  \\
		\log_r(\alpha) \geq \lfloor \log_r(\alpha) \rfloor \geq \lceil \frac{(n-1)}{r} \rceil,  \\
		\log_r(\alpha) \geq \lceil \frac{(n-1)}{r} \rceil,   \\
		\alpha \geq r^{\lceil \frac{n-1}{r} \rceil} 	.	  
		\eean
		Thus we have shown that if $\alpha < r^{k-1}$, we must have $\alpha \geq r^{\lceil \frac{n-1}{r} \rceil}$.  It follows that 
		\bean
		\alpha \geq \min \{r^{\lceil \frac{n-1}{r} \rceil},r^{k-1}\}.
		\eean		   
	\end{enumerate}    
\end{proof}

\subsection{Sub-Packetization Bound for MDS Codes with Optimal, Help-By-Transfer Repair of a Node Subset}	

\begin{cor} \label{cor:node_subset}  
	Let $\mathcal{C}$ be a linear $[n,k]$ MDS code over the vector alphabet $\fq^{\alpha}$ containing a distinguished set $W$ of $|W|=w \leq (n-1)$ nodes.  Each node in $W$ can be repaired, through linear repair, by accessing and downloading, precisely $\beta=\frac{\alpha}{r}$ symbols over \fq \  from each of the remaining $d=(n-1)$ nodes.
	In other words, the repair of each node in $W$ can be carried out through help-by-transfer with minimum repair bandwidth with only linear operations.   Then we must have 
	\bean
	\alpha \geq  \left\{ \begin{array}{rl} \min \{ r^{\lceil \frac{w}{r} \rceil},r^{k-1} \}, \  & w > (k-1) \\
		r^{\lceil \frac{w}{r} \rceil}, \ & w \leq (k-1) . \end{array} \right. 
	\eean
\end{cor}

\begin{proof}  We remark that even in this setting, it is known that $d\beta=d\frac{\alpha}{r}$ is the minimum repair bandwidth needed to repair the nodes in $W$, hence the nodes in $W$ are those for which the repair is optimal.  To prove the corollary, consider a subset $U \subseteq W$ and $|U| \leq k-1$ and repeat the steps used to prove Lemma \ref{lem:rsi} with this set $U$. Since in the proof of Lemma \ref{lem:rsi}, we only consider equations regarding repair of nodes in $U$, the proof of Lemma \ref{lem:rsi} will go through.   Assuming wolog that $n \notin W$, we will arrive at the following analogue of \eqref{eq:Cond_4}: 
	\bean
	\dim \left( \bigcap_{u \in U} <S_{n \rightarrow u}> \right) \leq \frac{\alpha}{r^{|U|}}.
	\eean
	For the case when $|W| > (k-1)$, we again extend the range of validity of this inequality to the case when $U$ is any subset of $W$, by first assuming that $\alpha < r^{k-1}$ and proceeding as in the proof of Theorem~\ref{thm:main} above.  For the case when $|W| \leq(k-1)$, no such extension is needed.  We then repeat the bipartite-graph-counting argument used in the proof of Theorem \ref{thm:main}, with the difference that the number of nodes on the right equals $w$ with $\{ S_{(n,j)}, j \in W\}$ as nodes in the right.   This will then give us the desired result.
\end{proof}

\subsection{Subpacketization Bound for All-Node Repair and Constant Repair Subspaces}	

\begin{cor} \label{sb_5} 
	Given a linear optimal access $\{(n,k,d),(\alpha,\beta)\}$ MSR code $\mathcal{C}$ with  $d =n-1$ and linear repair for all $n$ nodes with $S_{i \rightarrow j}=S_j$,$\forall i,j \in [n],i \neq j$ (thus the repair matrix $S_{i \rightarrow j}$ is independent of $i$), we must have:
	
	\bean
	\alpha \geq  \min \{ r^{\lceil \frac{n}{r} \rceil},r^{k-1}\}.
	\eean
	
\end{cor}

\begin{proof}
	For a given subset $U \subseteq [n]$ of nodes, of size $2 \leq |U| \leq k-1$, let $P \subseteq [n]$ be a second subset disjoint from $U$ (i.e., $ U \bigcap P = \emptyset$), of size $|P|=r$.    In this setting, the proof of Lemma \ref{lem:rsi} will go through for the pair of subsets $U,P$ and we will obtain that for any $p \in P$:
	\bean
	\dim \left( \bigcap_{u \in U} <S_{u}> \right) = \dim \left( \bigcap_{u \in U} <S_{p \rightarrow u}> \right) \leq \frac{\alpha}{r^{|U|}}
	\eean
	As before, we next extend the validity of the above inequality for any $U \subseteq [n]$ by assuming that $\alpha < r^{k-1}$ and following the same steps as in the proof of Theorem \ref{thm:main}.    Following this, we repeat the bipartite-graph construction and subsequent counting argument as in the proof of Theorem \ref{thm:main} with one important difference.   In the bipartitie graph constructed here, there are $n$ nodes on the right (as opposed to $(n-1)$), with $\{S_{1},\hdots,S_{n}\}$ as nodes in the right. The result then follows.
\end{proof}

\subsection{Sub-packetization Bound for an Arbitrary Number $d$ of Helper Nodes}
\begin{cor} \label{sb_4} 
	Let $\mathcal{C}$ be a linear $[n,k]$ MDS code over the vector alphabet $\fq^{\alpha}$ containing a distinguished set $W$ of $|W|=w \leq d$ nodes.  Each node $j$ in $W$ can be repaired, through linear repair, by accessing and downloading, precisely $\beta=\frac{\alpha}{s}$ symbols over \fq \, from each of the $d$ helper nodes where $s=d-k+1$, the $d$ helper nodes are any set of $d$ nodes apart from the failed node $j$.
	In other words, the repair of each node in $W$ can be carried through help by transfer with minimum repair bandwidth with only linear operations.   Then we must have 
	\bea
	\label{eq:subpkt_ws}\alpha \geq  \left\{ \begin{array}{rl} \min \{ s^{\lceil \frac{w}{s} \rceil},s^{k-1}\}, \  & w > (k-1) \\
		s^{\lceil \frac{w}{s}\rceil}, \ & w \leq (k-1) . \end{array} \right. 
	\eea

\end{cor}

\begin{proof}
	To prove this, we simply restrict our attention to the (punctured) code obtained by selecting a subset of nodes of size $n'=(d+1)$ that includes the subset $W$ for which optimal repair is possible.   Applying the results of Corollary~\ref{cor:node_subset} then gives us the desired result.
\end{proof}

\subsection{Subpacketization Bound for Arbitrary $d$ and Repair Subspaces that are Independent of the Choice of Helper Nodes}	

\begin{cor} \label{sb_6} 
	Let  $\mathcal{C}$ be a linear optimal-access $\{(n,k,d),(\alpha,\beta)\}$ MSR code for some $d$, $k \leq d \leq n-1$, and linear repair for all $n$ nodes.  We assume in addition, that every node $j \in [n]$ can be repaired by contacting a subset $D \subseteq [n]\setminus \{j\}$, $|D|=d$ of helper nodes in such a way that the repair matrix $S^D_{i \rightarrow j}$ is independent of the choice of the remaining $(d-1)$ helper nodes, i.e., $S^D_{i \rightarrow j}=S_{i \rightarrow j}$,$\forall j \in [n]$ and $i \in D$, $\forall D \subseteq [n]\setminus \{j\}, |D|=d$.  Then we must have:
	\bean
	\alpha \geq  \min \{ s^{\lceil \frac{n-1}{s} \rceil},s^{k-1} \}.
	\eean
	
\end{cor}

\begin{proof}
	Given a set $U, U \subseteq [n-1]$, of nodes of size $|U|$, $2 \leq |U| \leq k-1$, let us form a set $P$ of size $|P|=d-k+1$, such that $P \subseteq [n]$,  with $n \in P$, $ U \bigcap P = \emptyset$.   Let $V$ be any subset of $[n]\setminus \{U \cup P\}$ such that $|U \cup V \cup P|=d+1$.   Next, consider  the punctured code obtained by restricting attention to the node subset $\{U \cup V \cup P\}$.  The proof of Lemma \ref{lem:rsi} applied to the subset $\{U \cup V \cup P\}$ of nodes will then go through and we will obtain:
	\bean
	\dim \left( \bigcap_{u \in U} <S_{n \rightarrow u}> \right) \leq \frac{\alpha}{s^{|U|}}.
	\eean
	We then repeat the process of extending the above inequality for any $U \subseteq [n-1]$ by assuming $\alpha < s^{k-1}$ and following the proof of Theorem \ref{thm:main}.   Following this, we construct the bipartite graph as always and repeat the counting argument employed in the proof of Theorem \ref{thm:main} with one difference.   On the right side of the bipartite graph, we now have the $(n-1)$ repair matrices $S_{n \rightarrow 1},\hdots,S_{n \rightarrow n-1}$ as the right nodes of the bipartite graph.   This will give us the desired result.   We omit the details.
\end{proof}

\section{Structure of MDS codes with Optimal Access Repair of a subset of nodes, Optimal sub-packetization\label{sec:struct_thm}}
In this subsection, we are going to deduce the structure of linear MDS code that allows for optimal-access single-node repair for any node in a subset of $w \le k$ nodes with optimal sub-packetization level (i.e., achieving our lower bound on $\alpha$ shown in equation~\eqref{eq:subpkt_ws}). While deducing the structure we also assume that these codes have constant repair matrices that are independent of the helper node set i.e., $S^D_{i \rightarrow j}=S_j$ where $D$ is helper node set with cardinality $d$. Also, since the codes are optimal access, the rows of $(\beta \times \alpha)$ repair matrices $S_j$'s are comprised of standard basis vectors from $\fq^{\alpha}$. Note that the condition $S^D_{i \rightarrow j}=S_j$ is satisfied by a lot of constructions in literature.

Theorem~\ref{thm:repss_struct} completely describes the structure of repair matrices of the code whereas Theorem~\ref{thm:struct} describes the support of the generator matrix of the code in terms of these repair matrices. We present the statements of the Theorems first. Before proving them we present some Lemma's that prove some properties on the repair matrices.

\bthm[Repair Matrix Structure]\label{thm:repss_struct}
Let, $\calc$ be an $(n,k,d=k+s-1)$ MDS code with sub-packetization level $\alpha=s^p$ such that any single node within the set of $[1:w]$ nodes can be repaired by optimal-access by downloading $\beta = s^{p-1}$ symbols from any $d$ helper nodes, where $w=sp \le k$, then the repair subspaces are forced to be:
\bea
\label{eq:rep_subspace}<S_{ys+x+1}> = \text{span} \{ \mathbf{e}_{z} \mid z \in [s^p], z_y = x \} \text{ for all } x \in [0,s-1], y \in [0, p-1]
\eea
upto a permutation over $[\alpha]$, where, $(z_0, \cdots, z_{p-1})$ is an $s$-array representation of $z-1$ for $z \in [s^p]$ and $\{e_z \mid z \in [s^p]\}$ is the set of standard basis vectors of $\fq^{\alpha}$.
\ethm

\begin{thm}[Generator Matix Support] \label{thm:struct}  Let $s = d-k+1 \le r, r=n-k, p$ be positive integers. Consider an $[n,k]$ MDS code over the vector alphabet $\fq^{\alpha}$ where $\alpha=s^p$. Let, $[1:w]$ be a set of $w=sp$ nodes such that any node in that set can be repaired, through linear repair, by accessing and downloading, precisely $\beta=\frac{\alpha}{s}$ symbols over \fq \  from each of the set of $d$ helper nodes. The set of $d$ helper nodes $D$ is such that , $|D|=d$ and $D \subseteq [n] \setminus \{u\}$. Let $d \geq k+1$, $k \geq 3$ and the generator matrix be of the form,
	\bea
	G = \left[
	\scalemath{1}{
		\begin{array}{cccc}
			I_{\alpha} & 0 & \hdots & 0 \\
			0 & I_{\alpha} & \hdots & 0 \\
			\vdots & \vdots & \vdots & \vdots \\
			0 & 0 & \hdots & I_{\alpha} \\
			A_{1,1} & A_{1,2} & \hdots & A_{1,k} \\
			A_{2,1} & A_{2,2} & \hdots & A_{2,k} \\
			\vdots & \vdots & \vdots & \vdots \\
			A_{r,1} & A_{r,2} & \hdots & A_{r,k} \\
		\end{array}
	}
	\right] \ \ \text{where } \underbrace{A_{i, j}}_{(\alpha \times \alpha)} = \left[\begin{array}{c}
		\mathbf{v}_{i,j}(1)^T \\
		\mathbf{v}_{i,j}(2)^T \\
		\vdots\\
		\mathbf{v}_{i,j}(\alpha)^T
	\end{array}\right]  \ \ \label{Gform_struct}
	\eea
	and $\mathbf{v}_{i,j}(z)$ for all $z \in [\alpha]$ are vectors in $\mathbb{F}_q^{\alpha}$.

	We define two sets:
	\bean
	N_j = \{ z \in [s^p] \ \mid \ \mathbf{e}_{z} \in <S_{j}> \}, \ \ L_{z} = \{ j \in [sp] \ \mid \ \mathbf{e}_{z} \in <S_{j}>\} \text{ for all } j \in [sp] \text{ and } z \in [s^p],
	\eean
	where $S_1, \cdots, S_{sp}$ are the repair matrices corresponding to repair of nodes $1, \cdots , sp$ respectively. It is clear to see that $j \in L_{z}$ iff $z \in N_j$. 
	
	\ben
	\item For $i \in [r], j \in [sp+1, k]$, the matrices $A_{i,j}$ are diagonal i.e., for all $z \in [\alpha]$, the $z$-th row,
	\bea
	\label{eq:row_sup2}
	\mathbf{v}_{i,j}(z) = a_{i,j, z}\mathbf{e}_{z},
	\eea
	where $a_{i,j,z}$ is an element in $\fq^*$\ .
	\item For $i \in [r], j \in [sp]$, the rows of matrices $A_{i, j}$,  satisfy following properties:
	\ben
	\item For all $z \notin N_j$
	\bea
	\label{eq:row_sup0}\mathbf{v}_{i,j}(z) = a_{i,j, z} \mathbf{e}_{z},
	\eea
	where $a_{i,j,z}$ is an element in $\fq^*$\ .
	\item For any $z \in N_j$, $i \in [r]$:
	\bea
	\label{eq:row_sup1}
	\text{Support}(\mathbf{v}_{i,j}(z)) &\subseteq& \cap_{j' \in L_z \setminus \{j\}} N_{j'}\\
	\nonumber \cup_{i \in D'}\text{Support}(\mathbf{v}_{i,j}(z)) &=& \cap_{j' \in L_z \setminus \{j\}} N_{j'} \text{ for any } D' \subseteq [r] \text{ such that } |D'| = s\\
	\label{eq:row_sup3} \cup_{i = 1}^r \text{Support}(\mathbf{v}_{i,j}(z)) &=& \cap_{j' \in L_z \setminus \{j\}} N_{j'}\\
	\nonumber |\cup_{i = 1}^r \text{Support}(\mathbf{v}_{i,j}(z))| &=& s
	\eea
	\item For any $z \ne z' \in N_j$
	\bean
	( \cup_{i = 1}^r \text{Support}(\mathbf{v}_{i,j}(z)) \cap 	( \cup_{i = 1}^r \text{Support}(\mathbf{v}_{i,j}(z'))  &=& \phi
	\eean
	
	\een 
	\een
\end{thm}

Before presenting the proofs for Theorems~\ref{thm:repss_struct} and \ref{thm:struct}, we present the following lemma that shows that the dimension of repair sub-spaces either reduces by a factor of $s=d-k+1$ or it is zero.

\blem\label{lem:intersect_card}
Let $s = d-k+1, p$ be positive integers. Consider an $[n,k]$ MDS code over the vector alphabet $\fq^{\alpha}$ where $\alpha=s^p$. Let, $[1:w]$ be a set of $w=sp \le (n-1)$ nodes such that any node $u \in W$ can be repaired, through linear repair, by accessing and downloading, precisely $\beta=\frac{\alpha}{s}$ symbols over \fq \  from each of the set of $d$ helper nodes. Then,
%
%
%
%
\bean
\text{dim}\left( \cap_{j \in J} <S_{j}>\right) = \begin{cases} s^{p - |J|} \text{ or } 0 & |J| \le p\\
	0 & |J| > p
\end{cases}
\eean
where $S_{j}$, $j \in [w]$ are the constant repair matrices.
\elem
\bprf
Let us recall the counting argument on the bipartite graph used in the proof of Theorem~\ref{thm:main}. Assume nodes $\mathbf{e}_1, \cdots, \mathbf{e}_{\alpha}$ on the left and $S_{1}, \cdots, S_{{w}}$ on the right as shown in Figure~\ref{fig:bipartite_wnodes}.
\begin{figure}[ht!]
	\bc
	\includegraphics[width=0.3\textwidth]{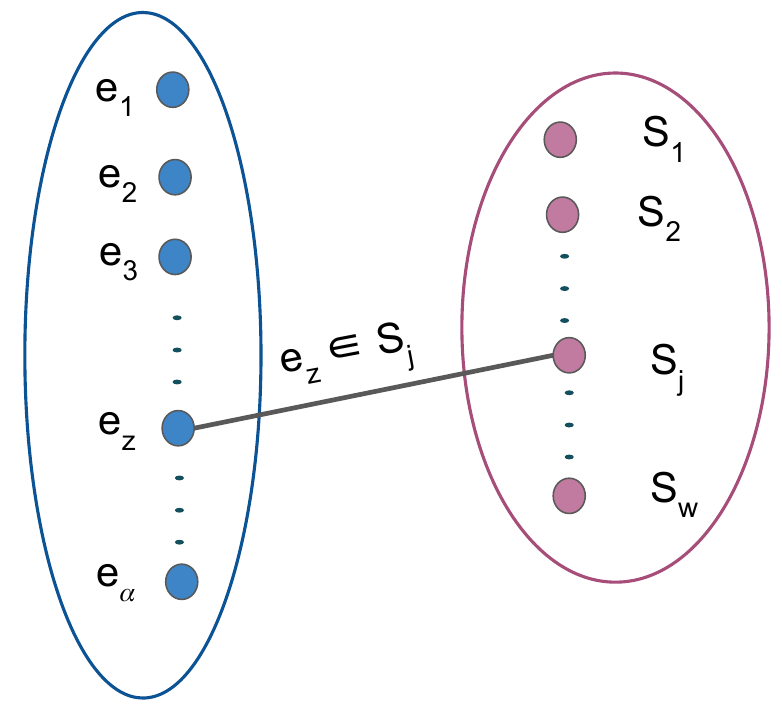}
	\caption{Bipartite graph indicating the elements in constant repair matrices \label{fig:bipartite_wnodes}}
	\ec
\end{figure}
The right degree of this bipartite graph is $\beta$ resulting in $w\beta$ edges. Let the average left degree be $d_{\text{avg}}$. Then: $d_{\text{avg}} \alpha = \beta w = sp\beta$, therefore $d_{\text{avg}} = p$.

We define the neighbors of left nodes $\mathbf{e}_{z}$ for $z \in [\alpha]$ in the graph as below:
\bean
L_{z} = \{ j \in [sp] \mid  \mathbf{e}_{z} \in <S_{j}>\}.
\eean 
From the proof of Corollary \ref{sb_6}, for any set $J \subseteq [w]$:
\bea
\label{eq:inter_subs_mds_repair} dim(\cap_{j \in J} <S_{j}> ) \leq \frac{\alpha}{s^{|J|}} .
\eea
By setting $J = L_z$ in the above equation where $z$ is a left node with maximum degree $d_{max} = |L_z|$ in the bipartite graph, we get that $d_{max} \le \log_s \alpha = p$. 
Hence we know that the degree of each left node is upper bounded by $p$. But since, the average left degree is $p$, it follows that the graph is regular with left degree $p$ and right degree $\beta$. Using \eqref{eq:inter_subs_mds_repair} for any set $J \subseteq [w]$:
\bea
\label{eq:ubcard}\text{dim}\left(\cap_{j \in J}<S_{j}>\right) &\le& \frac{\alpha}{s^{|J|}} = s^{p -|J|} (\text{as } \alpha = s^p) \\ 
\nonumber \implies \text{dim}\left(\cap_{j \in J}<S_{j}>\right) &=& 0 \text{ for } J \text{ such that } |J| > p
\eea
Let, $|J| \le p$, $\cap_{j \in J} <S_{j}> \ne \phi$, then there exists $z \in [\alpha]$ such that $\mathbf{e}_{z} \in \cap_{j \in J} <S_{j}>$. Consider the neighbors $L_{z}$ of $\mathbf{e}_{z}$ in bipartite graph. It is clear to see that $J \subseteq L_{z}$. We also have that:
\bean
\text{dim}(\cap_{j \in L_{z}} <S_{j}>) &\le& \frac{\alpha}{s^{|L_{z}|}} = \frac{s^p}{s^{|L_{z}|}} =  1 \ \ \ (\text{as } \alpha = s^p,|L_{z}| = p).
\eean
But, $\mathbf{e}_{z} \in \cap_{j \in L_{z}} S_{j}$ implying that:
\bean
\text{dim}(\cap_{j \in L_{z}} <S_{j}>) = 1 \ \text{ and } \cap_{j \in L_{z}} <S_{j}> = <\mathbf{e}_{z}>.
\eean
Applying Lemma \ref{lem:rsi} with $\ell = p \leq w$  for $d$ helper nodes  i.e., essentially replacing $r$ with $s$ in the proof of the lemma, it follows that:
\bea
\nonumber 1 = \text{dim}(\cap_{j \in L_{z}} <S_{j}>) & \le & \frac{1}{s^{p-|J|}} \text{dim}(\cap_{j \in J} <S_{j}>)\\
\label{eq:lbcard}\implies  \text{dim}(\cap_{j \in J} <S_{j}>) &\ge& s^{p -|J|}\\
\nonumber \implies  \text{dim}(\cap_{j \in J} <S_{j}>) &=& s^{p -|J|} \text{ (due to equations \eqref{eq:ubcard} and \eqref{eq:lbcard}).}
\eea
\eprf

\subsection{Proof of Theorem~\ref{thm:struct}, The Generator Matrix Support}
We will first show the matrices $A_{i, j}$ for $i \in [r]$, $j \in [sp+1:k]$ are diagonal matrices. We show this by proving that the support of $z$-th row is Support$(\mathbf{v}_{i,j}(z)) = \{z\}$. Firstly, from Lemma~\ref{lem:intersect_card} it is clear to see that:
	\bean
	\text{dim}(\cap_{j \in L_{z}} <S_{j}>)= \frac{\alpha}{s^{|L_{z}|}} = 1 \ \ \ (\text{as } |L_{z}| = p, \mathbf{e}_{z} \in \cap_{j \in L_{z}} <S_{j}>).
	\eean
	It also follows that $\cap_{j \in L_{z}} <S_{j}> = \mathbf{e}_{z}$. 
	From interference alignment conditions (see Lemma \ref{lem:rank_cond_p_node}) we have:
	\bean
	<S_{{j'}} A_{i, j}> &=& <S_{{j'}}> \ \ \text{ for all } i \in [r], \ \ j \in [k], j' \in [sp], j \ne j'.
	\eean
	For $j \in [sp+1, k]$ and $z \in [\alpha]$:
	\bean
	\cap_{j' \in L_{z}} < S_{{j'}} A_{i, j} > &=& \cap_{j' \in L_{z}} < S_{{j'}} >  = <\mathbf{e}_{z} > \ \text{ as } L_z \subseteq [1, sp], \ j \notin L_z \\
	< \left( \cap_{j' \in L_{z}} S_{{j'}} \right) A_{i, j} > &=& <\mathbf{v}_{i, j}(z)> = <\mathbf{e}_{z}>\\
	\implies \mathbf{v}_{i,j}(z) &=& a_{i,j, z} \mathbf{e}_{z} \text{ for all } i \in [r], j \in [sp+1, k], z \in [\alpha]
	\eean
	This implies that $A_{i,j}$ is diagonal for $i \in [r], j \in [sp+1:k]$. We will now show the support of rows of matrices $A_{i,j}$ for $i \in [r]$ and $j \in [sp]$. Let us consider the $z$-th row such that $z \in [\alpha] \setminus N_j$, then again from the interference alignment condition we have:
	\bean
	\cap_{j' \in L_{z}} < S_{{j'}} A_{i, j} > &=& \cap_{j' \in L_{z}} < S_{{j'}} >  = <\mathbf{e}_{z} > \ (\text{ as } j \notin L_{z})\\
	< \left( \cap_{j' \in L_{z}} S_{{j'}} \right) A_{i, j} > &=& <\mathbf{v}_{i, j}(z)> = <\mathbf{e}_{z}>\\
	\implies \mathbf{v}_{i,j}(z) &=& a_{i,j, z} \mathbf{e}_{z} \text{ for all } i \in [r], j \in [sp], z \notin N_j.
	\eean
	We therefore know the support of $\alpha - |N_j| = \alpha - \beta$ rows of $A_{i,j}$. We will now show the support structure for the remaining $\beta$ rows indexed by set $N_j$. For $j \in [sp]$, $z \in N_j$, it is clear that $j \in L_{z}$. By interference alignment condition (see Lemma \ref{lem:rank_cond_p_node}):
	\bean
	\cap_{j' \in L_{z} \setminus \{j\}} <S_{{j'}} A_{i, j}> &=& \cap_{j' \in L_{z} \setminus \{j\}} <S_{{j'}}> \text{ for all } i \in [r].
	\eean
	From Lemma \ref{lem:intersect_card}, 
	\bean
	\text{dim}(\cap_{j' \in L_{z} \setminus \{j\}} <S_{{j'}}>) &=& \frac{\alpha}{s^{|L_{z}|-1}} = s
	\eean
	as we know $\mathbf{e}_{z}$ is an element in  $\cap_{j' \in L_{z} \setminus \{j\}} <S_{{j'}}>$  by definition of $L_{z}$. 
	Using interference alignment conditions again we get:
	\bean
	\left( \cap_{j' \in L_{z} \setminus \{j\}} <S_{{j'}}> \right) A_{i, j} &=&  \cap_{j' \in L_{z} \setminus \{j\}} <S_{{j'}}> \\
	\implies \mathbf{v}_{i, j}(z) &\in &  \cap_{j' \in L_{z} \setminus \{j\}} <S_{{j'}}>
	\eean
	Let us consider repair of node $j$, where all the remaining systematic nodes $[k] \setminus \{j\}$ are part of helper node set and the rest $s$ nodes are picked from the set $P= \{k+1, \cdots, k+r\}$, then by the full rank condition from Lemma \ref{lem:rank_cond_p_node} we have:
	\bean
	\bigoplus_{i \in D'} <S_{j} A_{i, j}> = \fq^{\alpha}
	\eean
	where $D'$ is an $s$-element subset of $[r]$. This implies that:
	\bean
	\bigoplus_{i \in D'} <\{\mathbf{v}_{i,j}(z) \ |  \ z \in N_j \}> &=& \fq^{\alpha}\\
	\implies \cup_{i \in D', z \in N_j} \text{Support}(\mathbf{v}_{i,j}(z)) &=& [\alpha] 
	\eean
	We know that:
	\bean
	\mathbf{v}_{i,j}(z) & \in & \cap_{j' \in L_{z} \setminus \{j\}} <S_{{j'}}>\\
	\cup_{i \in D'} \text{Support}(\mathbf{v}_{i,j}(z)) &\subseteq& \cap_{j' \in L_{z} \setminus \{j\}} N_{j'} \text{ for any } D' \subseteq  [r], |D'|=s,\\
	\cup_{ z \in N_j} \cup_{i \in D'} \text{Support}(\mathbf{v}_{i,j}(z)) &\subseteq& \cup_{z \in N_j}  \left( \cap_{j' \in L_{z}  \setminus \{j\}} N_{j'} \right)\\
	\alpha = |\cup_{i \in D', z \in N_j} \text{Support}(\mathbf{v}_{i,j}(z))| &\le& |N_j|s = \alpha.
	\eean
	This implies that:
	\bean
	\cup_{i \in D'} \text{Support}(\mathbf{v}_{i,j}(z)) &=& \cap_{j' \in L_{z} \setminus \{j\}} N_{j'}, \text{ for any } D' \subseteq  [r], |D'|=s\\
	\text{Support}(\mathbf{v}_{i,j}(z)) &\subseteq& \cap_{j' \in L_{z} \setminus \{j\}} N_{j'}, \text{ for all } i \in [r]
	\eean
	and that,
	\bean
    \text{Support}(\mathbf{v}_{i,j}(z)) \cap \text{Support}(\mathbf{v}_{i,j}(z')) = \phi \text{ for any } z, z' \in N_j, \ z \ne z', i \in [r].
	\eean
%

\eprf

In this Lemma we look at some necessary conditions on the generator matrix coefficients. One of these conditions follow due to the MDS property of the code.
\blem
Let us consider the setting of Theorem~\ref{thm:struct}. The following properties are also necessary along with the support structure presented in Theorem~\ref{thm:struct}.
\ben
\item $a_{i,j,z}$ is diagonal element of the matrix $A_{i,j}$ for $i \in [r]$ and $j \in [sp+1:k]$. For every $z \in [\alpha]$, the matrix
	\bea
	\label{eq:qz}Q_{z} &=& \left[\begin{array}{cccc}
		a_{1,sp+1,z} & a_{1, sp+2, z} & \cdots & a_{1, k, z}\\
		a_{2,sp+1,z} & a_{2, sp+2, z} & \cdots & a_{2, k, z}\\
		\vdots & \vdots & \ddots & \vdots\\
		a_{r,sp+1,z} & a_{r, sp+2,z} & \cdots & a_{r, k, z}\\
	\end{array}\right],
	\eea
	is of dimension $r \times (k-sp)$ and has the property that every square sub-matrix is invertible.
	\item  For any $z \in N_j$, the following matrix:
		\bean
		\left[ \mathbf{v}_{1,j}(z) \ \mathbf{v}_{2,j}(z) \ \cdots, \ \mathbf{v}_{r,j}(z) \right]
		\eean
		when restricted to non-zero rows results in the generator matrix of an $[r,s]$ MDS code over $\mathbb{F}_q$.
\een
\elem
\bprf
Let us assume that $E$ is a set of $|E|=r$ erasures such that $E \cap [sp] = \phi$, ie., all the erasures seen in the first $k$ nodes belong to the set $[sp+1:k]$. Let $E \cap [sp+1:k] = \{j_1, \cdots, j_{\ell}\}$ and $[k+1:n] \setminus E = \{i_1, \cdots, i_{\ell} \}$ where $\ell \le \min\{r, k-sp\}$. Then to be able to recover from these $E$ erasures it is necessary that the $\ell \alpha \times \ell \alpha$ submatrix is invertible:
	\bean
	\left[\begin{array}{ccc}
		A_{i_1-k, {j_1}} & \cdots & A_{i_1-k, {j_{\ell}}}\\		
		\vdots & \ddots & \vdots\\
		A_{i_{\ell}-k, {j_1}} & \cdots & A_{i_{\ell}-k, {j_{\ell}}}\\
	\end{array}\right]
	\eean
	needs to be invertible for the MDS property to hold. The block matrices above are all diagonal, therefore the above matrix is invertible iff for every $z \in [\alpha]$ the matrix:
	\bean \left[\begin{array}{ccc}
		a_{i_1-k, j_1,z} & \cdots & a_{i_1-k, j_{\ell},z}\\		
		\vdots & \ddots & \vdots\\
		a_{i_{\ell}-k, j_1,z} & \cdots & a_{i_{\ell}-k, j_{\ell},z}\\		
	\end{array}\right] \text{ is invertible},
	\eean
	implying that any square submatrix of $(r \times (k-sp))$ martix $Q_{z}$ shown in equation~\eqref{eq:qz} is invertible.
\eprf
\subsection{Proof of Theorem~\ref{thm:repss_struct}, The Repair Matrix Structure}

In the following lemma we will show that the $w = sp$ repair subspaces can be grouped in to $p$ sets with each set containing $s$ repair subspaces. Within each group any two repair subspaces are disjoint and the direct sum of $s$ repair subspaces within a any group is equal to $\fq^{\alpha}$.

\blem\label{lem:p_partitions}
 Let $s = d-k+1\le r,   r=n-k, p$ be positive integers. Consider an $[n,k]$ MDS code over the vector alphabet $\fq^{\alpha}$ where $\alpha=s^p$. Let, $[1:w]$ be a set of $w=sp \le k$ nodes such that any node in this set can be repaired, through linear repair, by accessing and downloading, precisely $\beta=\frac{\alpha}{s}$ symbols over \fq \  from each of the set of $d$ helper nodes. The dimension of intersection of $j$ repair sub-spaces with each picked from a different group is given by $s^{p-j}$.

Let, $S_{j}$, $j \in [w]$ be the repair matrices corresponding to the $w$ nodes that can be repaired optimally, then there exists a partitioning $P_1, \cdots, P_p$ of the set $[sp]$ such that $|P_m| = s$ for all $m \in [p]$ and,
\bea
\label{eq:part_int1} \oplus_{j \in P_m} <S_{j}> &=& \fq^{\alpha} \text{ for any } m \in [p],\\ 
\label{eq:part_int} \text{dim}(\cap_{j \in J} <S_{j}>) &=& s^{p-|J|} \text{ for any } J \subseteq [sp] \text{ such that } |J| \le p, |J \cap P_m| \le 1 \text{ for all } m \in [p].
\eea
\elem

In order to prove this lemma we will now introduce some definitions.

\bdefn[Partitioning of R]
We say that the collection of repair matrices $S_1, \cdots, S_s$ partition a subspace $R$ if the following condition holds:
\bean
\oplus_{i=1}^s \left( <S_i> \cap \ R \right) = R.
\eean
\edefn

\emph{Proof of Lemma~\ref{lem:p_partitions}:}\\
Let us consider an arbitrary $[z \in \alpha]$ and let
\bean
L_{z} = \{ j \in [sp] \mid e_{z} \in <S_j>\} = \{j_1, \cdots, j_p\}.
\eean
as from Lemma~\ref{lem:intersect_card}, it follows that $|L_z|=p$. We now define subspace $R_y$ for any $y \in [p]$ as shown below:
\bean
R_y = \begin{cases}
\cap_{i=1}^{p-y} <S_{j_i} > & y < p\\
\fq^{\alpha} & \text{otherwise}.
\end{cases}
\eean
We will now come up with a set $P_1$ containing $s$ elements and show that $\{S_i \mid i \in P_1 \}$ partition the space $\fq^{\alpha}$. We do this inductively by first showing that $\{S_i \mid i \in P_1 \}$ partitions $R_1$ and then by showing that partitioning $R_{y}$ implies partitioning of $R_{y+1}$ for any $y \in [p-1]$. Since $R_p = \fq^{\alpha}$ this gives the desired result. We will later show that we can get $s$ element sets $P_2, \cdots, P_p$ that satisfy this property as well.

We will now proceed to identify $P_1$ and show that $\{S_i \mid i\in P_1\}$ partitions $R_1$. Looking at the bipartite graph (see Fig.~\ref{fig:bipart_r1}) whose left nodes are $\{e_{z'} \mid z' \in R_1\}$ and the right nodes are 
\bean \left\lbrace S_{i} \mid i \in P_1 \right\rbrace, \ \text{ where }  P_1 = \left\lbrace i \in [sp] \setminus \{j_1, \cdots, j_{p-1}\} \ \mid \ <S_{i}> \cap R_1 \ne \phi \right\rbrace. \eean
\begin{figure}[ht!]
	\bc
	\includegraphics[width=0.3\textwidth]{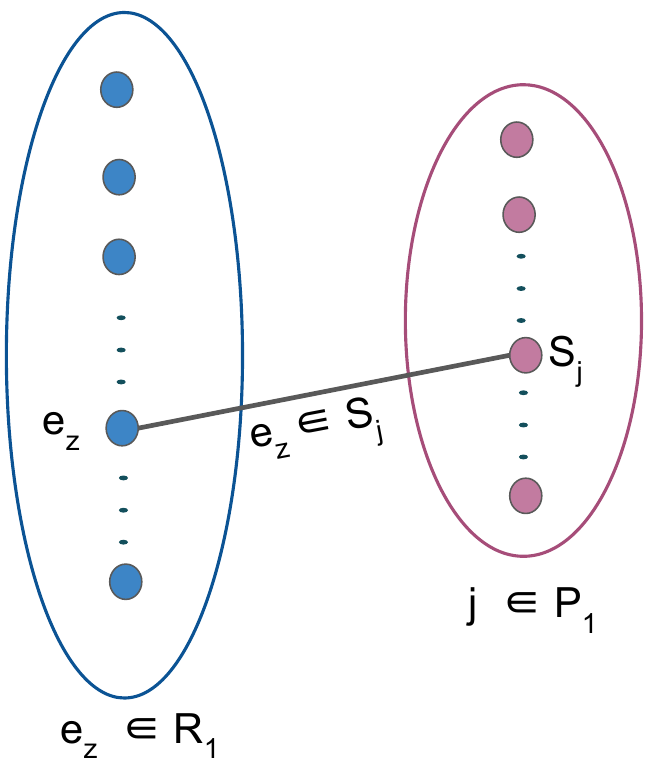}
	\caption{\label{fig:bipart_r1}$\{S_i \mid i \in P_1\}$ partitions subspace $R_1$}
	\ec
\end{figure}
It is clear that the number of left nodes is equal to $s$ as $\mathbf{e}_z \in R_1$ from Lemma~\ref{lem:intersect_card}. The left degree of node $z' \in R_1$ is given by $|L_{z'} \cap P_1|$. By the definition of $R_1$, $\{ j_1, \cdots, j_{p-1} \} \subseteq L_{z'}$ but $\{ j_1, \cdots, j_{p-1} \} \cap P_1 = \phi$. Therefore the left degree is equal to $|L_{z'}| - (p-1) = 1$ by Lemma~\ref{lem:intersect_card}. The right degree for any $i \in P_1$ is given by:
\bean
\text{dim} \left( <S_{i}> \cap \ R_1 \right) &=& 1 \text{ by Lemma} \ref{lem:intersect_card}.
\eean
Therefore the number of right nodes is equal to $|P_1| = |R_1| = s$. It now follows that:
\bean
\oplus_{i \in P_1} \left( <S_{i}> \cap \ R_1 \right) &=&  R_1
\eean
Therefore $\{S_{i} \mid i \in P_1 \}$ \emph{partitions} the space $R_1$. We will now show that if  $\{S_{i} \mid i \in P_1 \}$ partitions the space $R_y$ for any $y \in [p-1]$, it partitions the larger space $R_{y+1}$. By that assumption, we have that:
\bean
\oplus_{i \in P_1} \left( <S_{i}> \cap \ R_y \right) = R_y.
\eean
Let us define 
\bean 
N_i &=& \{ z \in [\alpha] \mid \mathbf{e}_z \in <S_i>\} \text{ for all } i \in [sp]\\
I_y &=& \{ z \in [\alpha] \mid \mathbf{e}_z \in R_y \} \text{ for all } y \in [p].
\eean
From the definition of $R_y$ it is clear to see that:
\bean
I_y = \begin{cases}
	\cap_{i=1}^{p-y} N_{j_i} & y < p\\
	[\alpha] & y = p,
 \end{cases}
\eean
and $|I_y| = s^y$. From \eqref{eq:row_sup3} in Theorem\ref{thm:struct} we have that:
\bea
\nonumber \cup_{i_0=1}^r  \text{Support}(\mathbf{v}_{i_0,j_{p-y}}(z)) &=& \cap_{j \in L_z \setminus \{ j_{p-y}\}} N_j \text{ for any } z \in N_{j_{p-y}}.
\eea
Notice that $I_y \subseteq N_{j_{p-y}}$. Therefore, for any $i \in P_1$ from the above equation it follows that:
\bea
\nonumber \cup_{z \in N_{i} \cap I_y }  \cup_{i_0=1}^r \text{Support}(\mathbf{v}_{i_0,j_{p-y}}(z)) &=& \cup_{z \in N_{i} \cap I_y } \left( \cap_{j \in L_z \setminus \{ j_{p-y} \}} N_j \right)\\
\label{eq:prop1}&\subseteq& N_{i} \cap \left(\cap_{i=1}^{p-y-1} N_{j_i}\right) = N_i \cap I_{y+1},
\eea
$\text{as } \{i\} \cup \{j_1, \cdots, j_{p-y}\} \subset L_z \text{ for } z \in N_i \cap I_y$. But we know that cardinality of the LHS is $|N_{i} \cap I_y|*s = s*s^{y-1}=s^y$ as the sets $\cup_{i_0=1}^r \text{Support}(\mathbf{v}_{i_0,j_{p-y}}(z))$ are disjoint across different $z \in N_{j_{p-y}}$ from Theorem~\ref{thm:struct}. This is same as the cardinality of RHS from Lemma~\ref{lem:intersect_card}. Therefore from equation \eqref{eq:prop1} we have,
\bea
\nonumber \cup_{z \in N_i \cap I_y} \cup_{i_0=1}^r \text{Support}(\mathbf{v}_{i_0,j_{p-y}}(z)) &=& N_{i} \cap I_{y+1} \text{ for all } i \in P_1.\\
\nonumber \cup_{i \in P_1 } \cup_{z \in N_{i} \cap I_y } \cup_{i_0=1}^r \text{Support}(\mathbf{v}_{i_0,j_{p-y}}(z)) &=& \cup_{i \in P_1} N_{i} \cap I_{y+1}\\
\label{eq:prop2}&\subseteq& I_{y+1}.
\eea
By the induction assumption the set $\{S_i \mid i \in P_1\}$ partitions $R_y$. Therefore for any $i, i' \in P_1$, 
\bean
\left(<S_i> \cap R_y\right) \cap \left( <S_{i'}> \cap R_y \right) &=& \phi\\
\implies ( N_i \cap I_{y} ) \cap ( N_{i'} \cap I_{y} ) &=& \phi.
\eean
Therefore the cardinality of LHS in equation~\eqref{eq:prop2} is $|P_1|*|N_i \cap I_y|*s = s^{y+1}$ due to Theorem~\ref{thm:struct}. This is equal to the cardinality of $I_{y+1}$ in the RHS from Lemma~\ref{lem:intersect_card}. Therefore from equation \eqref{eq:prop2} we have that,
\bean
\cup_{i \in P_1} (N_{i} \cap I_{y+1})
&=& I_{y+1}\\
\oplus_{i \in P_1} \left( <S_{i}> \cap R_{y+1} \right) &=&  R_{y+1}.
\eean
It is therefore clear that,  the collection $\{S_{i} \mid i \in P_1\}$ partitions the set $\fq^{\alpha}$ by setting $y=p-1$. We will now prove that there exist $p$ such sets by induction. Let us assume that we have $y$ such sets $P_1, P_2 \cdots, P_y$ such that for any $y' \in [y]$ 
\bean
\oplus_{i \in P_{y'}} <S_{i} > = \fq^{\alpha} 
\eean
then we will show that there is a set $P_{y+1}$ such that $\{S_i \mid i \in P_{y+1}\}$ partitions $\fq^{\alpha}$. 

We will first show that if one element each is picked from each partition they have a non-zero intersection i.e.,
\bean
\cap_{y'=1}^y <S_{i_{y'}}> \ne \phi \text{ for any } i_{y'} \in P_{y'}.
\eean
For $y = 1$, there trivially holds true. Assume that it is true for $y=t-1$:
\bean
\cap_{y'=1}^{t-1} N_{i_{y'}} &\ne& \phi  \text{ for any } i_{y'} \in P_{y'}\\
\implies  |\cap_{y'=1}^{t-1} N_{i_{y'}}| &=& s^{p-t+1} (\text{by Lemma \ref{lem:intersect_card}})
\eean
Looking at the bipartite graph with left nodes as elements in $\{ e_{z'} \mid z' \in \cap_{y'=1}^{t-1} N_{i_{y'}} \}$ and right nodes as elements in
\bean
\left\lbrace S_{j} \mid <S_{j}> \cap \left(\cap_{y'=1}^{t-1}<S_{i_{y'}}>\right) \ne \phi, \ j \in [sp] \setminus \{ i_1, i_2, \cdots, i_{t-1} \} \right\rbrace.
\eean
Number of left nodes is equal to $s^{p-t+1}$. The left degree of this graph is $p-t+1$ and the right degree is $s^{p-t}$. Therefore the number of right nodes is given by:
\bean
\frac{1}{s^{p-t}}(p-t+1)s^{p-t+1} &=& s(p-t+1).
\eean
It is clear to see that 
\bean
<S_{i}> \cap \left(\cap_{y'=1}^{t-1} <S_{i_{y'}}>\right) &=& \phi \text{ for any } i \in P_1 \cup P_2 \cdots \cup P_{t-1} \setminus \{i_1, \cdots, i_{t-1}\} .
\eean
Therefore we have $(s-1)(t-1)$ repair matrices that are disjoint with $ \cap_{y'=1}^{t-1} <S_{i_{y'}}>$ and none of them can be part of the right nodes of the bipartite graph. Therefore all the remaining nodes $sp-(t-1)-(s-1)(t-1)=s(p-t+1)$ that are not part of the $t-1$ disjoint sets $P_1, \cdots, P_{t-1}$ have non zero intersection. Given we have any $s$ element set $P_t$ disjoint with $P_1, \cdots, P_{t-1}$, it follows that: 
\bean
\cap_{y'=1}^{t} <S_{i_{y'}}> &\ne& \phi \text{ for any } i_{y'} \in P_{y'}
\eean
implying the equation~\eqref{eq:part_int} in the lemma statement.
We will use this to construct the set $P_{y+1}$ given there are sets $P_1, \cdots, P_{y}$ such that for any $y' \in [y]$, $\{S_i \mid i \in P_{y'}\}$ partitions $\fq^{\alpha}$. Let $\mathbf{e}_{z} \in \cap_{y'=1}^y S_{i_{y'}}$ such that $i_{y'} \in P_{y'}$ and $L_{z} = \{j_1, j_2 \cdots, j_p\}$ where $j_{y'} = i_{y'}$ for all $i \in [y]$. Looking at the bipartite graph with left nodes as elements from $R_1 = \cap_{i=1}^{p-1} <S_{{j_i}}>$ and right nodes as elements in 
\bean
\left\lbrace S_{i} \mid  i \in P_{y+1} \right\rbrace, \ \text{ where } P_{y+1} = \left\lbrace j \in [sp] \setminus \{j_1, \cdots, j_{p-1}\} \mid <S_{j}> \cap \ R_1 \ne \phi \right\rbrace.
\eean
Following a similar argument that shows that $\{S_i \mid i\in P_1\}$ partitions $\fq^{\alpha}$ it can be shown that $\{S_i \mid i\in P_{y+1}\}$ partitions $\fq^{\alpha}$. It can be verified that the $P_{y+1}$ that is picked is disjoint from the sets $P_1, \cdots, P_{y}$. By setting $y = p-1$, we have $p$ sets $P_1, \cdots, P_p$ that satisfy the properties stated in the theorem.
\eprf

\ \\
\bprf of Theorem~\ref{thm:repss_struct}
We can assume without loss of generality that the $p$ partitions are defined by $P_y = \{ys+x+1 \mid x \in \Zs\}$ $\{S_{ys+x+1} | x \in \Zs\}$ form a partition $P_y$. By varying $y \in [0,p-1]$, there are $p$ partitions. From Lemma \ref{lem:p_partitions} it follows that:
\bean
<S_{y s + x + 1}>  \cap <S_{y s + x' + 1}> &=& \phi \text{ for any } x \ne x' \in [0:s-1] \text{ from equation~\eqref{eq:part_int1}}
\eean
Let $z, z' \in [s^p]$ such that $z \ne z'$ where $(z_0, \cdots, z_{p-1})$ and $(z_0',\cdots, z_{p-1}')$ are $s$-array representations of $z-1$, $z'-1$ respectively, then it follows that:
\bean
\text{dim} (\cap_{y=0}^{p-1} <S_{y s + z_y + 1}>) &=& 1 \text{ from equation~\eqref{eq:part_int}}\\
\left(\cap_{y=0}^{p-1} <S_{y s + z_y' + 1}>\right) \cap \left(\cap_{y=0}^{p-1} <S_{y s + z_y + 1}>\right) &=& \cap_{y=0}^{p-1} (<S_{y s + z_y' + 1}> \cap <S_{y s + z_y + 1}> )\\
&=& \phi \text{ as there is a } y \in [0:p-1] \text{ such that } z_y \ne z_y'\\
\implies \cup_{z \in [s^p]} \left(\cap_{y=0}^{p-1} <S_{y s + z_y + 1}>\right) &=& \mathbb{F}_q^{\alpha}
\eean
Let $N_j = \{z \mid \mathbf{e}_z \in <S_j>\}$. We can define a permutation $\pi : [\alpha] \rightarrow [\alpha]$ as follows: 
\bean
\pi(\cap_{y=0}^{p-1} N_{ys+z_y+1}) = z
\eean
for any $z \in [s^p]$. Applying this permutation over the elements seen in the repair matrices, we get the repair subspace assignment shown in \eqref{eq:rep_subspace}.
\eprf
%
\section{ Deduced Structure from Theorems~\ref{thm:struct} and \ref{thm:repss_struct}\label{sec:clay}}
In this section we describe how Theorems~\ref{thm:struct} and \ref{thm:repss_struct} determine the support structure of parity-check equations that define the optimal access MDS codes with optimal sub-packetization level. We first describe the structure for an example and then provide a general description.

\subsection{Deduced structure of MDS code with optimal-access repair of $w=s$ nodes, where $\alpha=s$}

We restrict our attention to the case when, $d=n-1$, $w=r=s$ and $s \leq k$. From Theorem \ref{thm:struct}, we get the general form of 
$\{A_{p_i,u_j}\}$. Since the generator matrix is systematic, we can directly write down the parity-check matrix.  In the following, we show that this parity-check matrix can be written such that the construction can be viewed to have a coupled-layer structure, that is present in the \cite{YeBarg_2017},\cite{SasVajKum},\cite{LiTangTian}. We do note however, that we are only considering the case where $w$ nodes are repaired, not all the nodes as in the case of \cite{YeBarg_2017, SasVajKum}. For simplicity we illustrate the connection via an example. A general description of the support structure for $w=sp$ case where $p \le \lfloor \frac{k}{s}\rfloor$ is explained in next sub-section. Let $w=s=3,r=3,k=3$, so that $\alpha=s=3$. Let $U=[1,k]$ and $P=[k+1,n]$.\\
From Theorem~\ref{thm:repss_struct} it follows that:
\bean
S_j = <\mathbf{e}_j> \text{ for all } j \in [3].
\eean
Now from Theorem~\ref{thm:struct}, we obtain that for any $i, j, z \in [3]$:
\bean
\mathbf{v}_{i,j}(z) &=& \begin{cases}
	[v_{i,j,1}, \ v_{i,j,2}, \ v_{i,j,3}]^T & z = j \\
	a_{i,j,z}\mathbf{e}_{z} & z \ne j
\end{cases} 
\eean
Therefore for any $i \in [3]$:
\bea
A_{i,1}=\left[
\begin{array}{ccc}
	v_{i,1,1} & v_{i,1,2} & v_{i,1,3} \\
	0 & a_{i,1,2} & 0 \\
	0 & 0 & a_{i,1,3} \\
\end{array} 
\right], \ A_{i,2}=\left[
\begin{array}{ccc}
	a_{i,2,1} & 0 & 0\\
	v_{i,2,1} & v_{i,2,2} & v_{i,2,3} \\
	0 & 0 & a_{i,2,3}
\end{array} 
\right],\ A_{i,3}=\left[
\begin{array}{ccc}
	a_{i,3,1} & 0 & 0\\
	0 & a_{i,3,2} & 0 \\
	v_{i,3,1} & v_{i,3,2} & v_{i,3,3} 
\end{array} 
\right].\ \label{P9}
\eea
The parity check matrix is given by:
\bea
H = \left[
\begin{array}{ccc|ccc}
	A_{1,1} & A_{1,2} & A_{1,3} & -I_3 & & \\
	A_{2,1} & A_{2,2} & A_{2,3} & & -I_{3} & \\
	A_{3,1} & A_{3,2} & A_{3,3} & & & -I_3 \\
\end{array} 
\right]. \ \ \label{P10}
\eea
%

Upon substituting for the $A_{i,j}$ for $i \in [3], j \in [3]$, we obtain:
\bean
H = \left[
\scalemath{0.8}{
	\begin{array}{ccc|ccc|ccc|c|c|c}
		v_{1,1, 1}& v_{1,1, 2} & v_{1,1, 3}& a_{1,2,1} &  &  & a_{1,3,1}  &  &   & \multirow{3}{*}{$-I_3$} & &\\
		& a_{1,1,2} &  & v_{1,2, 1} & v_{1,2,2} & v_{1,2,3} &  &  a_{1,3,2}  &  &  & & \\
		& & a_{1,1,3} &  &  & a_{1,2,3} & v_{1,3, 1} & v_{1,3, 2} & v_{1,3, 3}   & &  &\\
		\hline
		v_{2,1,1} & v_{2,1, 2} & v_{2,1, 3}& a_{2,2,1} &  &  & a_{2,3,1} &  &  &  & \multirow{3}{*}{$-I_3$} & \\
		& a_{2,1,2} &  & v_{2,2,1} & v_{2,2,2} & v_{2,2,3} &  &  a_{2,3,2} & &  &  &\\
		&  & a_{2,1,3} &  &  & a_{2,2,3} & v_{2,3,1} & v_{2,3,2} & v_{2,3,3} &  &   &\\
		\hline
		v_{3,1,1} & v_{3,1,2} & v_{3,1,3}& a_{3,2,1} &  & & a_{3,3,1} &  &  & & & \multirow{3}{*}{$-I_3$}\\
		& a_{3,1,2} &  & v_{3,2, 1} & v_{3,2,2} & v_{3,2,3} & &  a_{3,3,2} &  & &   &\\
		& & a_{3,1,3} &  &  & a_{3,2,3} & v_{3,3,1} & v_{3,3,2} & v_{3,3, 3} &  & &
	\end{array} 
}
\right].
\eean
$H$ is an $r \alpha \times n\alpha$ matrix with $r$ thick rows and $n$ thick columns, where each thick row has $\alpha$ rows and each thick column has $\alpha$ columns. The rows and columns of $H$ can be permuted to obtain $H_{\text{perm}}$ where the elements that appear in the first block of $H_{\text{perm}}$ are  obtained from the elements at first row, first column within each block of $H$. $H_{\text{perm}}$ has $\alpha$ thick rows and thick columns, where each thick row has $r$ rows and each thick column has $n$ columns. 

\bean
H_{\text{perm}} = \left[
\scalemath{0.8}{
	\begin{array}{cccc|cccc|cccc}
		v_{1,1, 1} & a_{1,2,1} & a_{1,3,1} & & v_{1,1, 2} & & &  & v_{1,1, 3} & &  & \\
		v_{2,1, 1} & a_{2,2,1} & a_{2,3,1} & -I_3 & v_{2,1, 2} &  & & & v_{2,1, 3}  & & &\\
		v_{3,1, 1} & a_{3,2,1} & a_{3,3,1} & & v_{3,1,2} & & &  & v_{3,1, 3} & &  &\\
		\hline
		& v_{1,2, 1} & & &  a_{1,1,2} & v_{1,2, 2} & a_{1,3,2} &  & & v_{1,2,3} & & \\
		& v_{2,2,1} & & &  a_{2,1,2} & v_{2,2, 2} & a_{2,3,2}& -I_3 & & v_{2,2, 3} & & \\
		& v_{3,2,1} & & & a_{3,1,2} & v_{3,2, 2} & a_{3,3,2} & & & v_{3,2,3} & &\\
		\hline
		& & v_{1,3, 1} & & & & v_{1,3, 2} & &  a_{1,1,3} & a_{1,2,3} & v_{1,3,3} &  \\
		& & v_{2,3,1}  & & & & v_{2,3,2} & &  a_{2,1,3} & a_{2,2,3}  & v_{2,3, 3} &  -I_3 \\
		& & v_{3,3,1} & & & & v_{3,3,2} & &  a_{3,1,3} & a_{3,2,3} & v_{3,3,3}  &
	\end{array} 
}
\right]
\eean
This structure was first observed in the MSR constructions in \cite{AgaSasKum} and \cite{SasAgaKum}. The construction in \cite{AgaSasKum}, allowed optimal access repair of systematic nodes and had optimal sub-packetization $\alpha=r^{\frac{k}{r}}$, whereas the construction in \cite{SasAgaKum} allowed optimal access repair of all the nodes with a sub-packetization of $\alpha=r^{\lceil\frac{n}{r}\rceil}$. However, both the constructions were not explicit and were defined over large field size and limited to the case $d = n-1$. In \cite{RawKoyVis_msr}, the authors extended the construction from \cite{SasAgaKum} to support any $d < n-1$, with sub-packetization $\alpha=s^{\lceil\frac{n}{s}\rceil}$. This construction is not explicit as well, but it is the first one to have the structure of an optimal sub-packetization MSR code for any $(n, k, d)$.
This is followed by an explicit MSR code called coupled layer code construction that was independently discovered in \cite{YeBarg_2017}, \cite{SasVajKum}, \cite{LiTangTian}. 

\subsubsection{Selecting coefficients for the coupled layer code} Let $P=\left[\begin{array}{ccc}\underline{p}_1 & \cdots & \underline{p}_k \end{array}\right]$ be an $r \times k$ matrix such that $\left[\begin{array}{cc}
P & I_r
\end{array}\right]$ forms the parity check matrix of an $[n, k]$ MDS code, $p_{i,j}$ be the element at $i$-th row and $j$-th column in $P$.
By letting $a_{i, j, z} = p_{i, j}$ and $v_{i, j, j}=p_{i, j}$, $v_{i, j, j'}(z) = \gamma p_{i, j'}$ when $j\ne j'$, we get:
\bean
H_{\text{perm}}=\left[
\begin{array}{cccc|cccc|cccc}
	\underline{p}_1 & {\color{red}\underline{p}_2} & {\color{blue}\underline{p}_3}  & I_3 & {\color{red}\gamma \underline{p}_2} & &  & & {\color{blue}\gamma \underline{p}_3}& & & \\
	& 	{\color{red}\gamma \underline{p}_1} &  & & {\color{red}\underline{p}_1} & \underline{p}_2 & {\color{green}\underline{p}_3} & I_3 & & {\color{green}\gamma \underline{p}_3} \\
	& & {\color{blue}\gamma \underline{p}_1} & &  & & {\color{green}\gamma \underline{p}_2}  & & {\color{blue}\underline{p}_1} & {\color{green}\underline{p}_2} & \underline{p}_3  & I_3\\
\end{array} 
\right]
\eean
The coupled nature of the construction is now apparent. Columns with same color correspond to coupled symbols. In this case, there are three pairs of coupled symbols. This construction will be elaborated in detail in the following section.
%
%


\subsection{Deduced structure of MDS code with optimal-access repair of $w=sp$ nodes, where $\alpha=s^p$ for $p \le \lfloor \frac{k}{s} \rfloor$}

Let the parameters of MDS code be $(n=st, k=s(t-1), d=n-1)$ with optimal sub-packetization $\alpha = s^p$ as the MDS code can repair any single node within the set of $w=sp$ nodes by optimal access, where $p \in [t-1]$. In this section, we will use the two-tuple $(x, y)$ to represent a node with index $ys+x+1$ for $x \in [0:s-1]$ and $y \in [0:t-1]$. We will from now on use $\{\mathbf{e}_{\uz} \mid  \uz \in \mathbb{Z}_s^p \}$ to represent the standard basis of $\fq^{\alpha}$.


We will now use the support structure of parity matrix derived in Theorem~\ref{thm:struct} along with the repair matrix structure forced due to Theorem~\ref{thm:repss_struct} to show that the resultant parity structure is same as that of the construction presented in \cite{LiTangTian}. 
Let $\left(\mathbf{C}(0,0), \cdots, \mathbf{C}(s-1,t-1)\right)$ be a codeword where $\mathbf{C}(x,y) \in \fq^{\alpha}$ and $\mathbf{C}(x,y) = (\cxyz \mid \uz \in \mathbb{Z}_s^p)$. The $s*\alpha$ parity symbols are defined by the generator matrix as:
\bean
\left[\begin{array}{c}
	\mathbf{C}(0,t-1)\\
	\mathbf{C}(1,t-1)\\
	\vdots\\
	\mathbf{C}(s-1,t-1)\\
\end{array}\right] &=& \underbrace{\left[\begin{array}{cccc}
	A_{1, 1} & A_{1, 2} & \cdots & A_{1, k}\\
	A_{2, 1} & A_{1, 2} & \cdots & A_{1, k}\\
	\vdots & \ddots & \ddots & \vdots\\
	A_{n-k, 1} & A_{n-k, 2} & \cdots & A_{n-k, k}\\
\end{array}\right]}_{(s \alpha \times k \alpha)} \underbrace{\left[\begin{array}{c}
	\mathbf{C}(0,0) \\
	\mathbf{C}(1,0) \\
	\vdots\\
	\mathbf{C}(s-1, t-2)
\end{array}\right]}_{(k\alpha \times 1)}.
\eean
Let $\mathbf{v}_{i,j}(\uz)$ be $\uz$-th row of matrix $A_{i, j}$ and $\uz(x \rightarrow z_y) = (z_0, \cdots, z_{y-1}, x, z_{y+1}, \cdots, z_{p-1})$. For all $\ell \in [0:s-1]$ and $\uz \in \mathbb{Z}_s^p$:
\bea
\nonumber C(\ell,t-1;\uz) &=& \sum\limits_{y=0}^{t-2} \sum\limits_{x=0}^{s-1} \mathbf{v}_{\ell, ys+x+1}^T (\uz) \mathbf{C}(x,y)\\
\nonumber &=&   \sum\limits_{y=0}^{p-1}  \sum\limits_{x=0}^{s-1} \mathbf{v}_{\ell, ys+x+1}^T (\uz) \mathbf{C}(x,y) +   \sum\limits_{y=p}^{t-2}\sum\limits_{x=0}^{s-1} \mathbf{v}_{\ell, ys+x+1}^T (\uz) \mathbf{C}(x,y)\\
\nonumber &=& \sum\limits_{y=0}^{p-1}\sum\limits_{x=0, x \ne z_y}^{s-1} \mathbf{v}_{\ell, ys+x+1}^T (\uz) \mathbf{C}(x,y) + \sum\limits_{y=0}^{p-1} \mathbf{v}_{\ell, ys+z_y+1}^T (\uz) \mathbf{C}(z_y,y) +\\ \nonumber &&   \sum\limits_{y=p}^{t-2}\sum\limits_{x=0}^{s-1} \mathbf{v}_{\ell, ys+x+1}^T (\uz) \mathbf{C}(x,y)\\
\label{eq:pc_sup}&=& \sum\limits_{y=0}^{p-1} \sum\limits_{x=0: x \ne z_y}^{s-1} X \cxyz + \sum\limits_{y=0}^{p-1} \sum\limits_{x =0}^{s-1} X \cxyzc+\sum\limits_{y=p}^{t-2}\sum\limits_{x=0}^{s-1} X \cxyz\\
\nonumber &=& \sum\limits_{y=0}^{t-2} \sum\limits_{x=0}^{s-1} X \cxyz + \sum\limits_{y=0}^{p-1} \sum\limits_{x=0: x \ne z_y}^{s-1} X \cxyzc
\eea
In \eqref{eq:pc_sup}, $X$ represents placement for a coefficient, and this equation as we will see, follows due to Theorem \ref{thm:struct}. Firstly, from the Theorem~\ref{thm:repss_struct}, the structure of repair matrices of MDS codes that allow for optimal-access repair of single node within set of $w=sp$ nodes, $[1:w]$ and have optimal sub-packetization level is given by:
\bean
<S_{ys+x+1}> = \text{span}\{\mathbf{e}_{\uz} \mid \uz \in \mathbb{Z}_s^p, z_y = x\}, \ \forall y \in [0:p-1], \forall x \in [0:s-1]. 
\eean
Let,
\bean
N_j &=& \{ \uz \in \mathbb{Z}_s^p \mid e_{\tiny \uz} \in S_j \} \text{ for all } j \in [sp]\\
\text{Therefore, } N_{ys+x+1} &=& \{\uz \in \mathbb{Z}_s^p \mid  z_y = x\} \text{ for any } x \in [0:s-1], y \in [0:p-1].
\eean
Secondly, from Theorem~\ref{thm:struct} it follows that:
\ben
\item For $y \in [0:p-1]$, when $x \ne z_y$, it is clear that $\uz \notin N_{ys+x+1}$, therefore from equation~\ref{eq:row_sup0}
\bean
\text{ Support}(\mathbf{v}_{\ell, ys+x+1} (\uz)) = \{\uz\} \text{ for all } \ell \in [0:s-1].
\eean
\item For $y \in [0:p-1]$, when $x = z_y$, it is clear that $\uz \in N_{ys+x+1}$, therefore from equation~\ref{eq:row_sup1}
\bean
\text{Support}(\mathbf{v}_{\ell, ys+z_y+1}) &\subseteq& \cap_{j \in L_{\tiny \uz} \setminus \{ys+z_y+1\}} N_j \text{ where } L_{\tiny \uz} = \{ j \mid  \uz \in N_j \} = \{y's+z_{y'}+1 \mid y' \in [0:p-1] \}\\
&=& \cap_{y' \in [0:p-1] \setminus \{y\}} N_{y's+z_{y'}+1} \\
&=& \left\lbrace \uz(x \rightarrow z_y) \mid x \in [0:s-1] \text{ for all } \ell \in [0:s-1]. \right\rbrace
\eean
\item For $y \in [p:t-2]$, the nodes $ys+x+1 \in [sp+1: k]$ for some $x \in [0:s-1]$, therefore from equation~\eqref{eq:row_sup2}, 
\bean
\text{Support}(\mathbf{v}_{\ell, ys+x+1} (\uz)) = \{\uz\} \text{ for all } \ell \in [0:s-1].
\eean
\een
This explains how the equation~\eqref{eq:pc_sup} is obtained.
\subsubsection*{Coefficient Assignment For Explicit Optimal-Access MDS Code from \cite{LiTangTian}}\ \\
For all $\ell \in [0:s-1]$, $\uz \in \mathbb{Z}_s^p$, let us the $ r \alpha = s*s^p $ parity check equations be defined as follows:
\bea
\label{eq:oamds_pc} C(\ell,t-1;\uz) &=& \sum\limits_{y=0}^{t-2} \sum\limits_{x=0}^{s-1} \theta_{(x,y), \ell} \cxyz + \sum\limits_{y=0}^{p-1} \sum\limits_{\substack{x \in [0:s-1],\\ x \ne z_y}} \gamma \theta_{(x,y), \ell}  \cxyzc,
\eea
where $\gamma \in \fq^*$ and $\gamma^2 \ne 1$ and $\theta_{(x,y),\ell}$ are chosen such that 
\bea
\label{eq:theta}\Theta = \left[\begin{array}{cccc|cccc}
	1 & & & & \theta_{(0, 0),0} & \theta_{(0,0), 1} & \cdots & \theta_{(0,0),s-1}\\
	& 1 & & & \theta_{(1, 0),0} & \theta_{(1,0), 1} & \cdots & \theta_{(1,0),s-1}\\
	& & \ddots & & \vdots & \vdots & \ddots & \vdots\\
	& & & 1 & \theta_{(s-1, t-2),0} & \theta_{(s-1,t-2), 1} & \cdots & \theta_{(s-1,t-2),s-1}\\
\end{array} \right]
\eea
forms the generator matrix of $[n=st, k=s(t-1)]$ MDS code over \fq with $q \ge n-1$.

\underline{Optimal-Access Repair Property} Let $(x_0, y_0)$ be the node to be repaired for some $x_0 \in [0:s-1]$ and $y_0 \in [0:p-1]$, then the helper information sent by the the helper node $(x,y) \ne (x_0, y_0)$ where $x \in [0:s-1], \ y \in [0:t-1]$ is given by:
\bean
\{ \cxyz \mid \uz \in \mathbb{Z}_s^p, z_{y_0} = x_0 \}
\eean
For $\ell \in [0:s-1]$, $\uz \in \mathbb{Z}_s^p$ such that $z_{y_0} = x_0$:
\bea
\nonumber \sum\limits_{y=0}^{t-2} \sum\limits_{x=0}^{s-1} \theta_{(x,y), \ell} \cxyz + \sum\limits_{y=0}^{p-1} \sum\limits_{\substack{x \in [0,s-1]\\ x \ne z_y}} \gamma \theta_{(x,y), \ell}  \cxyzc &=& C(\ell,t-1;\uz)\\
\label{eq:pc_sysmsr_repair} \theta_{(x_0,y_0), \ell} \cxyznot + \sum\limits_{\substack{x \in [0,s-1]\\ x \ne x_0}} \gamma \theta_{(x,y_0), \ell}  \cxyzcnot &=& \kappa^*,
\eea
where $\kappa^*$ can be computed from known repair symbols. This is clear because $\cxyz$ is known for all $(x,y) \ne (x_0, y_0)$ and for any $y \ne y_0$ the symbols $\{ \cxyzc \mid x \in [0:s-1] \setminus \{z_y\} \}$ are also known as they are part of the helper information. From \eqref{eq:pc_sysmsr_repair}, it is clear to see that by varying $\ell \in [0:s-1]$ we get:
\bean
\scalebox{0.9}{$
\left[\begin{array}{ccccccc}
\gamma \theta_{(0,y_0), 0} & \cdots & \gamma \theta_{(x_0-1, y_0), 0} & \theta_{(x_0, y_0),0} & \gamma \theta_{(x_0+1, y_0),0} & \cdots & \gamma \theta_{(s-1, y_0),0}\\
\gamma \theta_{(0,y_0), 1} & \cdots & \gamma \theta_{(x_0-1, y_0), 1} & \theta_{(x_0, y_0),1} & \gamma \theta_{(x_0+1, y_0),1} & \cdots & \gamma \theta_{(s-1, y_0),1}\\
\vdots & \ddots & \vdots & \vdots & \vdots & \ddots & \vdots\\
\gamma \theta_{(0,y_0), s-1} & \cdots & \gamma \theta_{(x_0-1, y_0), s-1} & \theta_{(x_0, y_0),s-1} & \gamma \theta_{(x_0+1, y_0),s-1} & \cdots & \gamma \theta_{(s-1, y_0),s-1}
\end{array}\right] \left[ \begin{array}{c}
C(x_0,y_0; \uz(0 \rightarrow z_{y_0}))\\
C(x_0,y_0; \uz(1 \rightarrow z_{y_0}))\\
\vdots\\
C(x_0,y_0; \uz(s-1 \rightarrow z_{y_0}))\\
\end{array}\right] = \kappa^*.$}
\eean
Therefore we can recover symbols $\{\cxyzcnot \mid \forall x \in [0:s-1] \}$. By varying  $\uz \in \mathbb{Z}_s^p$ such that $z_{y_0} = x_0$ we can recover all the symbols of node $(x_0, y_0)$. \\ \ \\
\underline{MDS property:} In order to prove the MDS property it is enough to show that the code can recover from any $n-k$ erasures. Let ${E}$ be any erasure pattern corresponding to erasure of $n-k=s$ nodes, we will show that they can be recovered. If ${E} = \{(x,t-1) \mid x \in \Zs \}$, then the lost symbols are parity symbols, therefore they can be clearly recovered from the generator matrix. Let $P$ be the set of erased parity nodes and ${E} \setminus P$ be the remaining nodes. Let $|P|=(s-\mu)$ then $|E \setminus P| = \mu$ and let $\ell \in \{ x \in [0:s-1] \mid  (x,t-1) \notin P \}$ then
\bea
\sum\limits_{y = 0}^{t-2} \sum\limits_{x=0}^{s-1} \theta_{(x,y),\ell} \cxyz + \sum\limits_{y = 0}^{p-1} \sum\limits_{\substack{x \in [0:s-1] \\ x \ne z_y}} \gamma \theta_{(x,y),\ell} \cxyzc &=& C(\ell,t-1;\uz)\\
\label{eq:pc_mds_1}\sum\limits_{(x,y) \in {E} \setminus P}  \theta_{(x,y),\ell} \cxyz + \sum\limits_{\substack{y \in [0:p-1]\\ (z_y,y) \in {E}}} \sum\limits_{\substack{x\in [0:s-1] \\ x \ne z_y}} \gamma \theta_{(x,y),\ell} \cxyzc &=& \kappa^*
\eea
where $\kappa^*$ can be computed from unerased symbols. It is clear to see that if $|\{y \in [0,p-1]: (z_y,y) \in {E} \}| = 0$, then the $\mu$ symbols $\{\cxyz \mid (x,y) \in {E}\setminus P \}$ are the only unknowns in equation~\eqref{eq:pc_mds_1} and they can be recovered by using $\mu$ equations corresponding to $\ell \in \{ x \in [0:s-1] \mid (x, t-1) \notin P \}$. 
Let us define intersection score of an index $\uz \in \mathbb{Z}_s^p$ as:
\bean
\sigma(E, \uz) &=& |\{y \in [0:p-1]: (z_y,y) \in E\}| 
\eean
It is clear that the symbols $\{\cxyz \mid (x,y) \in {E} \setminus P, \uz \in \mathbb{Z}_s^p, \sigma(E, \uz) = 0 \}$ can be recovered. We will now show a sequential decoding algorithm where in the $j$-th step the symbols $\{ C(x,y;\uz) \mid \uz \in \mathbb{Z}_s^p, \sigma(E, \uz) = j\}$ are recovered. Let us assume that at the $j$-th step the erased symbols shown below are recovered
\bean
\{ C(x,y; \uz) \mid \uz \in \mathbb{Z}_s^p, \sigma(E, \uz) < j \},
\eean
we will now show that we can recover the following erased symbols:
\bean
\{ C(x,y; \uz) \mid \uz \in \mathbb{Z}_s^p, \sigma(E, \uz) = j \}.
\eean
%
Let $\sigma( E, \uz) = j$, $(z_{y_0}, y_0) \in E$ and let $\uw = \uz(x \rightarrow z_{y_0})$, where $x \ne z_{y_0}$, then:
\bean
\sigma({E}, \uw) &=& \sum\limits_{y=0}^{p-1} \mathbf{1}_{(w_{y},y) \in E} = \sum\limits_{y=0}^{p-1} \mathbf{1}_{(z_{y},y) \in E} -1 + \mathbf{1}_{(x,y_0) \in E}\\
&\le& \sigma(E, \uz)=j \text{ equality holds iff } (x,y_0) \in E.
\eean
Therefore the equation \eqref{eq:pc_mds_1}, for $(z_y, y) \in E$, $\cxyzc$ is already recovered iff $(x,y_0) \notin {E}$.  Therefore this equation simplifies to:
\bean
\sum\limits_{(x,y) \in {E} \setminus P}  \theta_{(x,y),\ell} \cxyz + \sum\limits_{\substack{y \in [0,p-1] \\ (z_y,y) \in {E}}} \sum\limits_{\substack{x\in [0,s-1], x \ne z_y \\ (x,y) \in {E}}} \gamma \theta_{(x,y),\ell} \cxyzc &=& \kappa^*.\\
\eean
Let $E_2(\uz) = \{(x,y) \in E \mid y \in [0:p-1], x \ne z_y, (z_y, y) \in {E} \}$, then the equation above further simplifies to:
\bean
\sum\limits_{(x,y) \in E_2(\uz)}  \theta_{(x,y),\ell} \left(\cxyz + \gamma \cxyzc\right) + \sum\limits_{(x,y) \in {E} \setminus (P \cup E_2(\uz))} \theta_{(x,y),\ell} \cxyz  &=& \kappa^*\\
\eean
By using $\mu$ equations corresponding to $\ell \in [0:s-1]$ such that $(\ell,t-1) \notin P$ we can recover $\mu$ symbols given below:
\bean
\{ \cxyz \mid (x,y) \in {E} \setminus \left( P \cup E_2(\uz) \right) \} \cup \{\cxyz + \gamma \cxyzc \mid (x,y) \in E_2(\uz) \}
\eean
At the end of step $j$, for any $\uz$ such that $\sigma(E, \uz) = j$ and $(x,y) \in E_2(\uz)$, we have access to $\cxyz + \gamma \cxyzc$ and $\cxyzc + \gamma \cxyz$ as $(z_y, y) \in E_2(\underline{z}(x \rightarrow z_y))$ and $\sigma(E, \uz(x \rightarrow z_y)) = j$. We can therefore recover the symbols $\cxyz, \cxyzc$ for $(x,y) \in E_{2}(\uz)$. Therefore at the end of step $j$, we have recovered:
\bean
\{\cxyz \mid (x,y) \in {E} \setminus P, \uz \in \mathbb{Z}_s^{p}, \sigma(E, \uz) = j \}. 
\eean
Therefore, by induction we have proved that the symbols:
\bean
\{\cxyz \mid  (x,y)\in E \setminus P, \uz \in \mathbb{Z}_s^p \} 
\eean
are recovered. The erased parity symbols in $P$ can be recovered from the remaining symbols.

\begin{note}
	Note that when $p=t-1$, this results in optimal access systematic MSR code with parameters $(n=st, k=s(t-1), d=n-1)$ that has optimal sub-packetization level $\alpha = s^{t-1}$.
\end{note}

\subsection{Support of Optimal Sub-Packetization Level MSR Construction}\label{sec:supp_msr}

The optimal-access MDS code defined by the parity check equations in equation~\eqref{eq:oamds_pc} can be extended to optimal access MSR construction with parameters $(n=st, k=s(t-1), d=n-1, \alpha = s^t)$ presented in \cite{YeBarg_2017, SasVajKum, LiTangTian} using the parity-checks shown below:
\bean
\label{eq:oamsr_pc} \sum\limits_{y=0}^{t-1} \sum\limits_{x=0}^{s-1} \theta_{(x,y), \ell} ( \cxyz + \mathbf{1}_{x \ne z_y} \gamma \cxyzc ) = 0 \text{ for all } \ell \in [0:s-1], \uz \in \Zst
\eean
where
\bean
\left[\begin{array}{cccc}
	\theta_{(0, 0),0} & \theta_{(1,0), 0} & \cdots & \theta_{(s-1,t-1),0}\\
	\theta_{(0, 0),1} & \theta_{(1,0), 1} & \cdots & \theta_{(s-1,t-1),1}\\
	 \vdots & \vdots & \ddots & \vdots\\
\theta_{(0, 0),s-1} & \theta_{(1,0), s-1} & \cdots &\theta_{(s-1,t-1),s-1}\\
\end{array} \right]
\eean
is a $(s \times st)$ parity-check matrix of $[n=st, k=s(t-1)]$ MDS code.

\bibliographystyle{IEEEtran}
\bibliography{msrjournal}

\begin{thebibliography}{10}
\providecommand{\url}[1]{#1}
\csname url@samestyle\endcsname
\providecommand{\newblock}{\relax}
\providecommand{\bibinfo}[2]{#2}
\providecommand{\BIBentrySTDinterwordspacing}{\spaceskip=0pt\relax}
\providecommand{\BIBentryALTinterwordstretchfactor}{4}
\providecommand{\BIBentryALTinterwordspacing}{\spaceskip=\fontdimen2\font plus
\BIBentryALTinterwordstretchfactor\fontdimen3\font minus
  \fontdimen4\font\relax}
\providecommand{\BIBforeignlanguage}[2]{{%
\expandafter\ifx\csname l@#1\endcsname\relax
\typeout{** WARNING: IEEEtran.bst: No hyphenation pattern has been}%
\typeout{** loaded for the language `#1'. Using the pattern for}%
\typeout{** the default language instead.}%
\else
\language=\csname l@#1\endcsname
\fi
#2}}
\providecommand{\BIBdecl}{\relax}
\BIBdecl

\bibitem{BalKum}
S.~B. Balaji and P.~V. Kumar, ``A tight lower bound on the sub-packetization
  level of optimal-access msr and mds codes,'' 2018.

\bibitem{VajBalKum}
M.~Vajha, B.~S. Babu, and P.~V. Kumar, ``Explicit msr codes with optimal
  access, optimal sub-packetization and small field size for $d = k+1, k+2,
  k+3$,'' 2018.

\bibitem{SasAgaKum}
B.~Sasidharan, G.~K. Agarwal, and P.~V. Kumar, ``A high-rate {MSR} code with
  polynomial sub-packetization level,'' in \emph{Proc. {IEEE} International
  Symposium on Information Theory, {ISIT}}, 2015, pp. 2051--2055.

\bibitem{SasVajKum}
B.~Sasidharan, M.~Vajha, and P.~V. Kumar, ``An explicit, coupled-layer
  construction of a high-rate {MSR} code with low sub-packetization level,
  small field size and all-node repair,'' \emph{CoRR}, vol. abs/1607.07335,
  July 2016.

\bibitem{YeBarg_2017}
M.~Ye and A.~Barg, ``Explicit constructions of optimal-access mds codes with
  nearly optimal sub-packetization,'' \emph{IEEE Trans. Inf. Theory}, vol.~63,
  no.~10, pp. 6307--6317, Oct 2017.

\bibitem{RawKoyVis_msr}
A.~S. Rawat, O.~O. Koyluoglu, and S.~Vishwanath, ``Progress on high-rate {MSR}
  codes: Enabling arbitrary number of helper nodes,'' in \emph{2016 Information
  Theory and Applications Workshop, {ITA} 2016, La Jolla, CA, USA, January 31 -
  February 5, 2016}, 2016, pp. 1--6.

\bibitem{TamoWangBruck}
I.~Tamo, Z.~Wang, and J.~Bruck, ``Access versus bandwidth in codes for
  storage,'' \emph{IEEE Trans. Inf. Theory}, vol.~60, no.~4, pp. 2028--2037,
  April 2014.

\bibitem{VajRamPur_Clay}
M.~Vajha, V.~Ramkumar, B.~Puranik, G.~R. Kini, E.~Lobo, B.~Sasidharan, P.~V.
  Kumar, A.~Barg, M.~Ye, S.~Narayanamurthy, S.~Hussain, and S.~Nandi, ``Clay
  codes: Moulding {MDS} codes to yield an {MSR} code,'' in \emph{Proc. 16th
  {USENIX} Conference on File and Storage Technologies, Oakland, CA, USA},
  2018, pp. 139--154.

\bibitem{BalNikVaj}
B.~S. Babu, M.~N. Krishnan, M.~Vajha, V.~Ramkumar, B.~Sasidharan, and P.~V.
  Kumar, ``Erasure coding for distributed storage: an overview,'' \emph{Sci.
  China Inf. Sci.}, vol.~61, no.~10, pp. 100\,301:1--100\,301:45, 2018.

\bibitem{RasShaKum_pm}
K.~V. Rashmi, N.~B. Shah, and P.~V. Kumar, ``{Optimal Exact-Regenerating Codes
  for Distributed Storage at the MSR and MBR Points via a Product-Matrix
  Construction},'' \emph{IEEE Trans. Inf. Theory}, vol.~57, no.~8, pp.
  5227--5239, Aug. 2011.

\bibitem{PapDimCad}
D.~Papailiopoulos, A.~Dimakis, and V.~Cadambe, ``{Repair Optimal Erasure Codes
  through Hadamard Designs},'' \emph{IEEE Trans. Inf. Theory}, vol.~59, no.~5,
  pp. 3021--3037, 2013.

\bibitem{TamWanBru}
I.~Tamo, Z.~Wang, and J.~Bruck, ``Zigzag codes: {MDS} array codes with optimal
  rebuilding,'' \emph{IEEE Trans. Inf. Theory}, vol.~59, no.~3, pp. 1597--1616,
  2013.

\bibitem{WanTamBru_allerton}
{Wang, Z. and Tamo, I. and Bruck, J.}, ``{On Codes for Optimal Rebuilding
  Access},'' in \emph{{Proc. IEEE 47th Annual Allerton Conference on
  Communication, Control, and Computing, 2009}}, pp. {1374--1381}.

\bibitem{CadJafMalRamSuh}
V.~Cadambe, S.~A. Jafar, H.~Maleki, K.~Ramchandran, and C.~Suh, ``Asymptotic
  interference alignment for optimal repair of mds codes in distributed
  storage,'' \emph{IEEE Trans. Inf. Theory}, vol.~59, no.~5, pp. 2974--2987,
  2013.

\bibitem{AgaSasKum}
G.~K. Agarwal, B.~Sasidharan, and P.~V. Kumar, ``{An alternate construction of
  an access-optimal regenerating code with optimal sub-packetization level},''
  in \emph{{National Conference on Communication (NCC)}}, 2015.

\bibitem{YeBar_1}
M.~Ye and A.~Barg, ``Explicit constructions of high-rate {MDS} array codes with
  optimal repair bandwidth,'' \emph{{IEEE} Trans. Information Theory}, vol.~63,
  no.~4, pp. 2001--2014, 2017.

\bibitem{LiTangTian}
J.~Li, X.~Tang, and C.~Tian, ``A generic transformation for optimal repair
  bandwidth and rebuilding access in mds codes,'' in \emph{2017 IEEE
  International Symposium on Information Theory (ISIT)}, June 2017, pp.
  1623--1627.

\bibitem{TamWanBru_access}
{I. Tamo and Z. Wang and J. Bruck}, ``{Access vs. bandwidth in codes for
  storage},'' in \emph{{IEEE International Symposium on Information Theory,
  ISIT 2012.}}\hskip 1em plus 0.5em minus 0.4em\relax {IEEE}, {2012}, pp. {1187
  -- 1191}.

\bibitem{GopTamCal}
S.~Goparaju, I.~Tamo, and A.~R. Calderbank, ``An improved sub-packetization
  bound for minimum storage regenerating codes,'' \emph{{IEEE} Trans. on Inf.
  Theory}, vol.~60, no.~5, pp. 2770--2779, 2014.

\bibitem{HuangParamXian}
K.~{Huang}, U.~{Parampalli}, and M.~{Xian}, ``Improved upper bounds on
  systematic-length for linear minimum storage regenerating codes,'' \emph{IEEE
  Transactions on Information Theory}, vol.~65, no.~2, pp. 975--984, 2019.

\bibitem{OmarGur}
O.~Alrabiah and V.~Guruswami, ``An exponential lower bound on the
  sub-packetization of msr codes,'' in \emph{Proceedings of the 51st Annual ACM
  SIGACT Symposium on Theory of Computing}, ser. STOC 2019.\hskip 1em plus
  0.5em minus 0.4em\relax Association for Computing Machinery, 2019, p.
  979–985.

\bibitem{ShahRashKumKann_ia}
N.~B. Shah, K.~V. Rashmi, P.~V. Kumar, and K.~Ramchandran, ``Interference
  alignment in regenerating codes for distributed storage: Necessity and code
  constructions,'' \emph{IEEE Trans. Inf. Theory}, vol.~58, no.~4, pp.
  2134--2158, April 2012.

\end{thebibliography}

\begin{appendices}
\label{appendix}
\section{Row Spaces}

\begin{lem} \label{Row_Spaces}
	Let $A,B$ be nonsingular $(\alpha \times \alpha)$ matrices and $\{P_1,P_2,Q_1,Q_2\}$ be matrices of size $(m \times \alpha)$.  Then if 
	\bean
	<P_1A> = <P_2B>, & & <Q_1A> = <Q_2B> ,
	\eean
	we have 
	\bean
	\left(<P_1> \bigcap <Q_1> \right) A  & = & 
	\left(<P_2> \bigcap <Q_2> \right) B. 
	\eean
	
	\label{lem:row space} 
\end{lem}

\begin{proof} 
	For the case when $A$ is nonsingular, we have that: 
	\bean
	\left(<P_1> \bigcap <Q_1> \right) A  & = & <P_1A>  \ \bigcap \ <Q_1 A>  .
	\eean
	The result then follows from noting that: 
	\bean
	\left(<P_1> \bigcap <Q_1> \right) A & = & <P_1A> \bigcap <Q_1A> \\
	\ = \ <P_2B> \bigcap <Q_2B> & = & \left(<P_2> \bigcap <Q_2> \right) B. 
	\eean
\end{proof}
\section{Condition for Repair of a systematic Node} 

\begin{lem} \label{lem:rank_cond_p_node}
	Let the linear $\{ (n,k,d=n-1), (\alpha,\beta), B, \mathbb{F}_q \}$ MSR code $\mathcal{C}$ be encoded in systematic form, where nodes $\{u_1, \cdots, u_k\}$ are the systematic nodes and nodes $\{p_1, \cdots, p_r\}$ are the parity nodes. Let the corresponding generator matrix $G$ be given by: 
	\bea
	G = \left[
	\scalemath{1}{
		\begin{array}{cccc}
			I_{\alpha} & 0 & \hdots & 0 \\
			0 & I_{\alpha} & \hdots & 0 \\
			\vdots & \vdots & \vdots & \vdots \\
			0 & 0 & \hdots & I_{\alpha} \\
			A_{p_1,u_1} & A_{p_1,u_2} & \hdots & A_{p_1,u_k} \\
			A_{p_2,u_1} & A_{p_2,u_2} & \hdots & A_{p_2,u_k} \\
			\vdots & \vdots & \vdots & \vdots \\
			A_{p_r,u_1} & A_{p_r,u_2} & \hdots & A_{p_r,u_k} \\
		\end{array} 
	}
	\right], \ \ \label{G_syst}
	\eea
	Let $\left\{ S_{(i,u_1)} \mid i \in [n]-\{u_1\}\right\}$ be the repair matrices	associated to the repair of the systematic node $u_1$. Then we must have that for $j \in [k]$, $j \neq 1$:
	\bea
	<S_{(u_j,u_1)}> = <S_{(p_1,u_1)} A_{p_1,u_j}>= \hdots = <S_{(p_r,u_1)} A_{p_r,u_j}> \label{eq:IA} 
	\eea
	and
	\bea
	rank \left( \left[ \begin{array}{c}  S_{(p_1,u_1)} A_{p_1,u_1}  \\ \vdots \\ S_{(p_r,u_1)} A_{p_r,u_1} \end{array} \right] \right) = \alpha \label{eq:FR} .
	\eea
\end{lem}
\begin{proof}
	Let $[m_1,...,m_k]$ be the $k$ message symbols (each message symbol $m_{j}$ being a $(1 \times \alpha)$ row vector) encoded by the MSR code i.e., the resultant codeword is given by $G[m_1,...,m_k]^T$. 
	For the repair of systematic node $u_1$, the data collected by the replacement node is given by:
	\bean 
	\left[ \begin{array}{c} S_{(u_2,u_1)} \uc{u_2}^T  \\ \vdots \\ S_{(u_k,u_1)} \uc{u_k}^T \\ S_{(p_1,u_1)} \uc{p_1}^T \\ \vdots \\ S_{(p_r,u_1)} \uc{p_r}^T \end{array} \right] 
	& = & 
	\left[ \begin{array}{c} S_{(u_2,u_1)} m^T_{2}  \\ \vdots \\ S_{(u_k,u_1)} m^T_{k} \\ S_{(p_1,u_1)} \left( \sum_{j=1}^k A_{p_1,u_j}m^T_j  \right) \\ \vdots \\ S_{(p_r,u_1)} \left( \sum_{j=1}^k A_{p_r,u_j}m^T_j\right) \end{array} \right] .
	\eean
	Let 
	\bean
	T_{u_1} & = \left[ \begin{array}{cccccc}T_{(u_2,u_1)}, & \cdots & T_{(u_k,u_1)}, & T_{(p_1,u_1)}, & \cdots & T_{(p_r,u_1)} \end{array} \right] 
	\eean
	be the $(\alpha \times (n-1)\beta)$ matrix used to derive the contents of the replacement node $u_1$, with each $T_{(i, j)}$ being an $(\alpha \times \beta)$ submatrix.  Then we must have that 
	\bean
	\left[ \begin{array}{cccccc}T_{(u_2,u_1)}, & \cdots & T_{(u_k,u_1)}, & T_{(p_1,u_1)}, & \cdots & T_{(p_r,u_1)} \end{array} \right] 	\left[ \begin{array}{c} S_{(u_2,u_1)} m^T_{2},  \\ \vdots \\ S_{(u_k,u_1)} m^T_{k} \\ S_{(p_1,u_1)} \left( \sum_{j=1}^k A_{p_1,u_j}m^T_j \right) \\ \vdots \\ S_{(p_r,u_1)} \left( \sum_{j=1}^k A_{p_r,u_j}m^T_j \right) \end{array} \right]
	& = & m^T_1 .
	\eean
	Since this must hold for all data vectors $m^T_j$, we can equate the matrices that premultiply  $m^T_j$ on the left, on both sides.  If we carry this out for $j=1$, we will obtain that: 
	\bean
	\sum_{i=1}^r T_{(p_i,u_1)}S_{(p_i,u_1)}A_{p_i,u_1} & = & I ,
	\eean
	which implies:
	\bean
	\left[ \begin{array}{ccc} T_{(p_1,u_1)} & \cdots & T_{(p_r,u_1)} \end{array} \right] \left[ \begin{array}{c}  S_{(p_1,u_1)} A_{p_1,u_1}  \\ \vdots \\ S_{(p_r,u_1)} A_{p_r,u_1} \end{array} \right]
	& = & I , 
	\eean
	which in turn, forces:
	\bea
	\text{rank} \left( \left[ \begin{array}{c}  S_{(p_1,u_1)} A_{p_1,u_1}  \\ \vdots \\ S_{(p_r,u_1)} A_{p_r,u_1} \end{array} \right] \right) = \alpha ,
	\eea
	and
	\bea
	\text{rank} \left( \left[ \begin{array}{ccc} T_{(p_1,u_1)} & \cdots & T_{(p_r,u_1)} \end{array} \right] \right) = \alpha . \label{Temp:FR}
	\eea
	The above equation proves \eqref{eq:FR}.
	For the case $j \neq 1$:
	\bean
	T_{(u_j,u_1)}S_{(u_j,u_1)} + \sum_{i=1}^r T_{(p_i,u_1)}S_{(p_i,u_1)}A_{p_i,u_j} & = & 0, 
	\eean
	and
	\bea
	\left[ \begin{array}{cccc}T_{(u_j,u_1)}, & T_{(p_1,u_1)}, & \cdots & T_{(p_r,u_1)} \end{array} \right] \left[ \begin{array}{c} S_{(u_j,u_1)} \\ S_{(p_1,u_1)} A_{p_1,u_j}  \\ \vdots \\ S_{(p_r,u_1)} A_{p_r,u_j}, \end{array} \right]
	& = & 0. \label{eq:IA_1}
	\eea
	Equations \eqref{eq:IA_1} and \eqref{Temp:FR} imply:
	\bea
	\text{rank} \left( \left[ \begin{array}{c} S_{(u_j,u_1)} \\ S_{(p_1,u_1)} A_{p_1,u_j}  \\ \vdots \\ S_{(p_r,u_1)} A_{p_r,u_j}, \end{array} \right] \right) \leq \beta. 
	\eea
	Given that $\text{rank}(S_{(p_i,u_1)} A_{p_i,u_j})=\beta$ for all $i \in [r]$ and $\text{rank}(S_{(u_j,u_1)})= \beta$, this implies
	\bea
	<S_{(u_j,u_1)}> = <S_{(p_1,u_1)} A_{p_1,u_j}>= \hdots = <S_{(p_r,u_1)} A_{p_r,u_j}> \label{IA:1}. 
	\eea
	The above equation proves \eqref{eq:IA}.
\end{proof}

\end{appendices}
\end{document}